\pdfoutput=1
\documentclass[acmsmall]{acmart}
\setcopyright{cc}
\setcctype{by}
\acmJournal{PACMMOD}
\acmYear{2026}
\acmVolume{4}
\acmNumber{3 (SIGMOD)}
\acmArticle{250}
\acmMonth{6}
\acmDOI{10.1145/3802127}



\received{October 2025}
\received[revised]{January 2026}
\received[accepted]{February 2026}

\usepackage[noabbrev,capitalize,nameinlink]{cleveref}
\usepackage{tikz}
\usepackage{enumitem}
\usepackage{amsfonts, amsmath, amsthm, mathtools}
\usepackage{thmtools,thm-restate}
\usepackage{algorithm}
\usepackage[noend]{algpseudocode}
\usepackage{algpseudocode}
\usepackage{cleveref}
\usepackage{listings}
\usepackage{xurl}
\usepackage{xspace}
\usepackage{makecell}
\usepackage{diagbox}
\usepackage{tabularx}
\usepackage{marginnote}
\usepackage{bm}
\usepackage[svgnames]{xcolor}
\usepackage[most]{tcolorbox}

\newtcolorbox{reviewercomment}{
    colback=gray!10,
    colframe=gray!40,
    boxrule=0pt,  
    arc=2pt,      
    left=4pt, right=4pt, top=4pt, bottom=4pt
}

\theoremstyle{plain}
\newtheorem{theorem}{Theorem}

\newtheorem{lemma}[theorem]{Lemma}

\DeclarePairedDelimiter{\ceil}{\lceil}{\rceil}
\DeclarePairedDelimiter{\floor}{\lfloor}{\rfloor}

\newcommand{\ul}[1]{\ }

\newcommand{\review}[2]{#2\xspace}

\newcommand{\algnamefull}{Workload-Aware Incremental Reclustering\xspace}
\newcommand{\algname}{WAIR\xspace}
\newcommand{\sysnamefull}{Automatic Reclustering Service\xspace}
\newcommand{\sysname}{ARS\xspace}

\newcommand{\zonemap}{zonemap\xspace}

\begin{document}


\renewcommand{\thefigure}{\arabic{figure}}
\renewcommand{\thealgorithm}{\arabic{algorithm}}
\setcounter{figure}{0}
\setcounter{algorithm}{0}
\setcounter{page}{1}

\title{Workload-Aware Incremental Reclustering in Cloud Data Warehouses}

\author{Yipeng Liu}
\affiliation{%
    \institution{Tsinghua University}
    \city{Beijing}
    \country{China}
}
\email{yipeng.liu@tuna.tsinghua.edu.cn}

\author{Renfei Zhou}
\orcid{0000-0002-1825-0097}
\affiliation{%
    \institution{Carnegie Mellon University}
    \city{Pittsburgh}
    \country{USA}
}
\email{renfeiz@andrew.cmu.edu}

\author{Jiaqi Yan}
\orcid{0000-0001-5109-3700}
\affiliation{%
    \institution{Snowflake Inc.}
    \city{San Carlos}
    \country{USA}
}
\email{jiaqi@snowflake.com}

\author{Huanchen Zhang}
\authornote{Huanchen Zhang is also affiliated with the Shanghai Qi Zhi Institute. Corresponding author.}
\affiliation{%
    \institution{Tsinghua University}
    \city{Beijing}
    \country{China}
}
\email{huanchen@tsinghua.edu.cn}

\begin{abstract}
    \phantomsection
    \label{sec:abstract}
    Modern cloud data warehouses store data in micro-partitions and rely on metadata (e.g., zonemaps) for efficient data pruning during query processing.
    Maintaining data clustering in a large-scale table is crucial for effective data pruning.
    Existing automatic clustering approaches lack the flexibility required in dynamic cloud environments with continuous data ingestion and evolving workloads.
    This paper advocates a clean separation between \textit{reclustering policy} and \textit{clustering-key selection}.
    We introduce the concept of boundary micro-partitions
    \marginnote{\ul{R1.M1}}[-1.2em] \review{a}{that sit on the boundary of query ranges.}
    We then present \algname, a workload-aware algorithm to identify and recluster only boundary micro-partitions most critical for pruning efficiency.
    \algname achieves near-optimal (with respect to fully sorted table layouts) query performance but incurs significantly lower reclustering cost with a theoretical upper bound.
    We further implement the algorithm into a prototype reclustering service and \marginnote{\ul{Meta}\\\ul{R1.O1}\\\ul{R4.O1}}[-1.2em]
    \review{m}{evaluate on standard benchmarks (TPC-H, DSB) and a real-world workload.}
    \marginnote{\ul{R4.O3}}
    \review{d}{Results show that \algname improves query performance and reduces the overall cost compared to existing solutions.}
\end{abstract}

\pagestyle{fancy}
\maketitle

\section{INTRODUCTION}

The rapid growth of data in modern enterprises has created significant challenges for managing and analyzing large datasets.
To handle this scale, organizations are progressively migrating their database systems from traditional on-premises infrastructure to cloud-based environments \cite{DBLP:journals/tkde/DongZLZ24}.
Cloud data warehouses, such as Snowflake \cite{DBLP:conf/sigmod/DagevilleCZAABC16,DBLP:conf/nsdi/VuppalapatiMATM20}, Redshift \cite{DBLP:conf/sigmod/ArmenatzoglouBB22}, and BigQuery \cite{DBLP:journals/pvldb/0001GLRSTVADMPS20}, have emerged as a popular solution.
They adopt an architecture that disaggregates compute from storage to provide outstanding elasticity and scalability.
A key feature of these cloud data warehouses is that they partition data into small, fixed-size units (we call these units ``micro-partitions'', following Snowflake's terminology) and store them in cloud object storage (e.g., AWS S3) using proprietary file formats or open-source ones such as Apache Parquet \cite{parquet}.
The system typically maintains the metadata (e.g., min-max \zonemap{s} \cite{DBLP:conf/vldb/Moerkotte98}) for each micro-partition in a separate service to enable efficient scan-set pruning during query optimization and execution.
Studies have shown that effective pruning is critical for improving the performance of analytical queries \cite{DBLP:conf/sigmod/DagevilleCZAABC16,DBLP:conf/sigmod/ZimmererDKWOK25}.

While micro-partitions are automatically created in data ingestion order, this natural ordering often does not align with the optimal physical layout for analytical workloads.
Maintaining optimal clustering is essential to ensure that only relevant micro-partitions are accessed during queries.
For example, consider a table storing sales information where data arrives chronologically (in batches) in $order\_timestamp$ order.
If the micro-partitions are clustered according to the data's natural ingestion order, queries that contain filter conditions on $customer\_id$ must access a large portion of these micro-partitions, each containing only a few records that match the filter conditions.

For large tables with continuous data ingestion, manually maintaining the clustered state places a huge burden on database users to monitor and analyze DML patterns and workload distributions.
Existing approaches, such as Qd-tree~\cite{DBLP:conf/sigmod/YangCWGLMLKA20} and MDDL in Redshift~\cite{DBLP:conf/sigmod/DingABPJPPPPPSS24}, assume a stationary workload and learn an optimal layout for the entire table from historical queries.
They are, however, suboptimal under continuous data ingestion and evolving workload patterns.
Iceberg~\cite{ibe}, Delta Lake~\cite{dlfrlc, dlddlc}, and Databricks~\cite{dblc, dbalc} allow dynamically switching clustering keys to accommodate workload shifts, but the updated physical layout is only applied to newly ingested data.
Snowflake~\cite{ckct} and Dremio~\cite{dmaic} propose the data-driven incremental reclustering, where they develop static metrics to evaluate the clustering quality of each micro-partition and only recluster the ill-clustered ones.
Although these approaches reduce the overall reclustering cost, they are unaware of workloads, resulting in unnecessary reclustering of rarely accessed data.

In this paper, we introduce \algnamefull (\textbf{\algname}), a reclustering algorithm (and service) that can efficiently and automatically maintain a high-quality clustered state for large, micro-partitioned tables in a cloud data warehouse under continuous data ingestion and evolving query patterns.
The key insight is that a small set of boundary micro-partitions (i.e., those that sit on the boundary of query ranges) largely determines pruning effectiveness.
Built upon the theoretical analysis of \algname's amortized cost, \algname adopts an incremental workload-aware approach that prioritizes boundary micro-partitions with higher expected payoff.
We develop a cost model for selecting a proper set of micro-partitions to recluster after each query to maximize the overall cost reduction (i.e., query execution savings - reclustering cost) within a time window.
By decoupling the reclustering policy from the clustering-key selection, \algname provides the flexibility for each micro-partition to choose its own clustering key to maximize the pruning gains.
We then implement an automatic reclustering service based on the \algname algorithm.
The service is completely off the query's critical path and is a drop-in component that leverages query-execution statistics to schedule and execute reclustering for a cloud data warehouse.
We experimentally compare \algname against a variety of baseline approaches, including research proposals and publicly documented commercial methods,
\marginnote{\ul{Meta}\\\ul{R1.O1}\\\ul{R4.O1}}[-1.2em]
\review{m}{using standard benchmarks (TPC-H \cite{tpch}, DSB \cite{DBLP:journals/pvldb/DingCGN21}) and a real-world workload.}
Our results show that \algname significantly improves query performance and lowers the overall operational cost simultaneously compared to the baselines.
\marginnote{\ul{Meta}\\\ul{R1.O4}\\\ul{R4.O1}}[-1.2em]
\review{m}{Sensitivity analysis using our workload generator shows that \algname remains robust under a wide range of workload scenarios.}

We make four primary contributions in this paper.
First, we propose a policy-key decoupled taxonomy and framework for reclustering.
Second, we introduce the concept of boundary micro-partitions that dominate scan-set pruning efficiency, and prove a theoretical upper bound for the cost of reclustering all these boundary micro-partitions.
Third, we propose \algname based on boundary micro-partitions that can automatically determine the reclustering timing and the micro-partition set to maximize overall cost reduction.
Finally, we implement a workload-aware automatic reclustering service based on \algname and demonstrate its advantages in cost, query performance, and robustness against diverse workloads over existing solutions.

\section{PRELIMINARIES}
\label{sec:preliminaries}

\begin{figure}[t!]
    \centering
    \includegraphics[width=\linewidth]{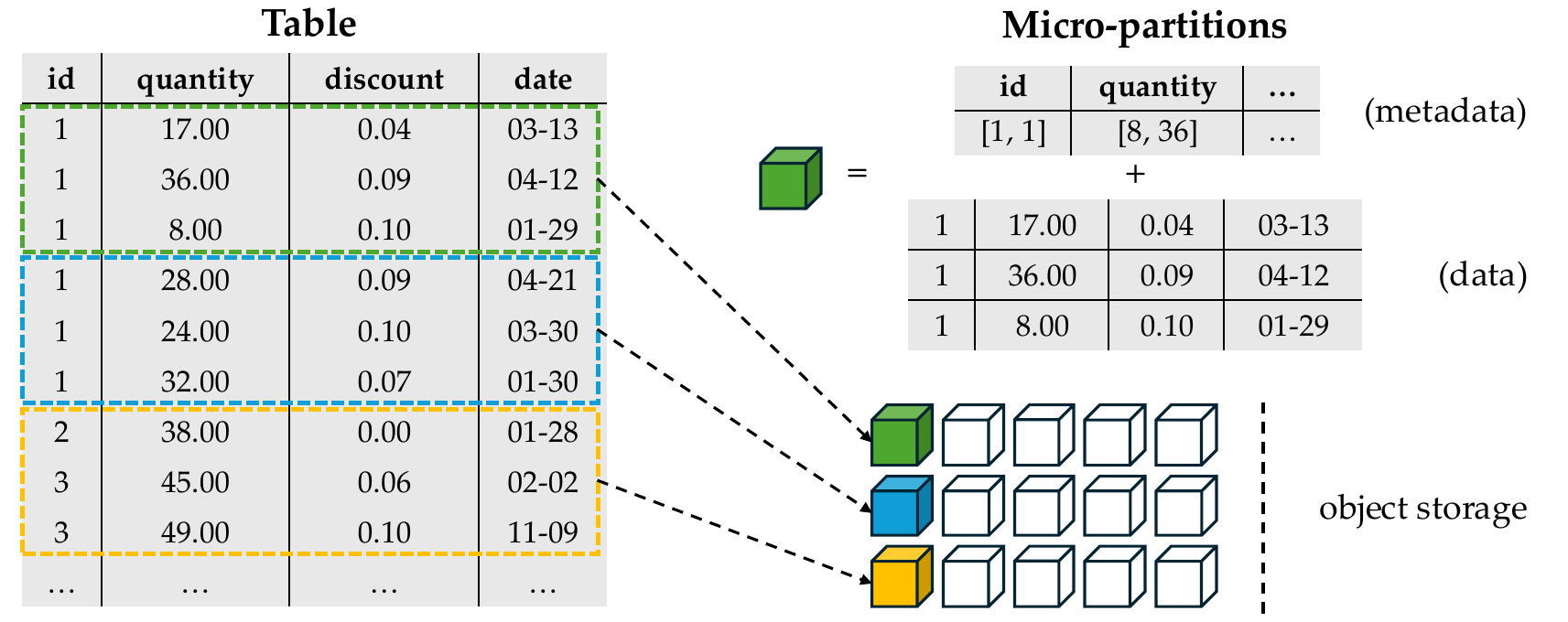}
    \caption{Micro-partitions in Cloud Data Warehouses. \textmd{A table is structured as multiple micro-partitions stored in the object storage layer. Each micro-partition comprises both metadata (e.g., zonemaps) and data blocks. During query execution, the metadata is used to prune irrelevant micro-partitions.}}
    \Description{Micro-partitions in Cloud Data Warehouses}
    \label{fig:micro-partitions}
\end{figure}

A defining feature of cloud data warehouses is the disaggregation of compute and storage for independent scaling \cite{DBLP:conf/cidr/Zaharia0XA21,DBLP:journals/pvldb/BarnhartBCHIJKKKMMNTUY24,DBLP:conf/sigmod/DagevilleCZAABC16}.
For queries on large tables, avoiding full table scans is critical for achieving optimal analytical performance.
A widely adopted approach is partitioning \cite{DBLP:conf/sigmod/LarsonCFHMNPPRRS13,DBLP:conf/sigmod/YangCWGLMLKA20,DBLP:conf/sigmod/DingMCWLLKGK21}.
In cloud data warehouses such as Snowflake, tables are divided into \textbf{micro-partitions}\footnote{Throughout this paper, the term ``partition'' typically refers to micro-partitions, unless otherwise specified.
} \cite{DBLP:conf/sigmod/DagevilleCZAABC16}, a concept similar to ``row groups'' in columnar data formats like Parquet \cite{parquet} and ORC \cite{orc}.
As illustrated in \cref{fig:micro-partitions}, incoming data is automatically partitioned into \textit{immutable} files (i.e., micro-partitions) according to its natural ingestion order.
These micro-partitions reside in cloud object storage (e.g., AWS S3, Azure Blob Storage) and serve as fundamental units for both data retrieval and pruning during query processing.
Each micro-partition is maintained at a fixed size, typically a few tens of megabytes in practice (e.g., 32MB) \cite{DBLP:journals/pvldb/DurnerL023}.

To facilitate scan-set pruning, the system maintains metadata such as \zonemap{s} \cite{DBLP:conf/vldb/Moerkotte98} to store the min/max values for the columns in each micro-partition.
During query execution, predicates (e.g., \texttt{WHERE $col$ BETWEEN $a$ AND $b$}) are evaluated against the \zonemap of each micro-partition.
If the partition's min-max range falls outside the query predicate, the entire partition can be safely skipped, thus drastically reducing the amount of data read from cloud storage.

The cloud's ability to provision nearly unlimited resources on demand has made \textit{cost} a first-class citizen in system optimization \cite{DBLP:conf/cidr/ZhangLY24}.
In most cloud services, compute costs are directly tied to the CPU times with negligible setup fees \cite{aeodp}.
For the reclustering task, the storage cost remains stable because the overall data volume does not vary much in different clustering states.
Costs generated by sending S3 requests are minor compared to the compute and storage costs \cite{asp,DBLP:conf/closer/RosatiFPTL18a}.
Data transfer within the same region (e.g., from S3 buckets to compute instances) is free \cite{asp}.
Therefore, the overall operational cost for query processing and reclustering is dominated by the consumed CPU time.

\subsection{Benefits and Challenges of Clustering}

\begin{figure}[t!]
    \centering
    \includegraphics[width=\linewidth]{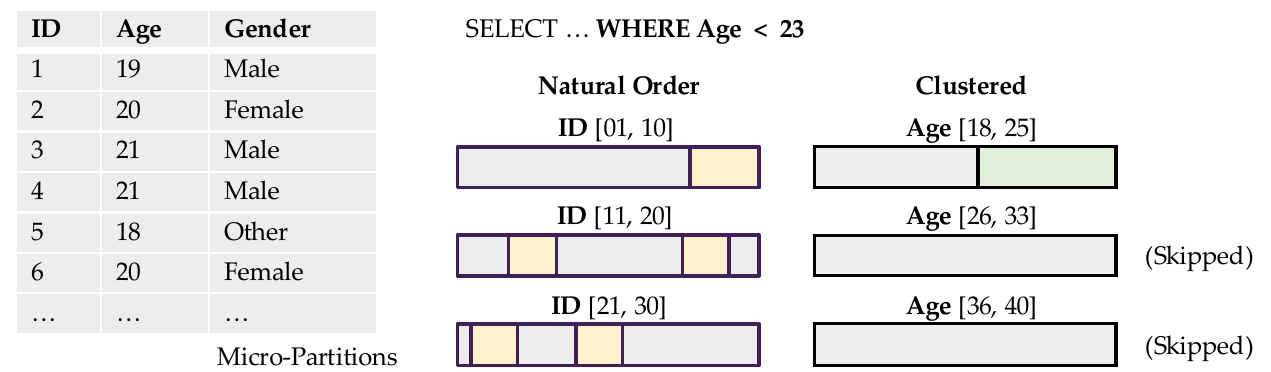}
    \caption{Benefits of Clustering. \textmd{The query predicated on \texttt{Age} must scan all micro-partitions in a naturally ordered table, while clustering by \texttt{Age} allows pruning all but the first micro-partition.}}
    \Description{Benefits of Clustering}
    \label{fig:clustering-benefit}
\end{figure}

By default, micro-partitions naturally cluster data based on its arrival order.
However, many queries include predicates that do not align with this order, making an alternative data clustering beneficial \cite{DBLP:conf/sigmod/BudalakotiZWKKW24,DBLP:conf/sigmod/DingABPJPPPPPSS24}.
A table is considered clustered if its micro-partitions are organized according to a specific order that aligns with common query patterns.
In this clustered state, rows that are logically close (e.g., having a small value distance in the specified column) are also stored physically close together, ideally within the same partition.

A key benefit of clustering is to improve the effectiveness of partition pruning.
As shown in \cref{fig:clustering-benefit}, after sorting the table by the \texttt{Age} column, the query filtering on \texttt{Age} only needs to scan one partition rather than three if the data were ordered differently by \texttt{ID}.
\marginnote{\ul{R1.O1}}[-1.2em]
\review{a}{Joins, LIMITs, and other SQL operations that benefit from partition pruning \cite{DBLP:conf/sigmod/ZimmererDKWOK25} naturally perform better when the table is clustered.}
Maintaining effective clustering is crucial to achieving optimal query performance, especially for large tables \cite{DBLP:conf/sigmod/ZimmererDKWOK25}.

Maintaining effective yet efficient clustering, however, presents significant challenges. Large tables in cloud data warehouses rarely remain static; they grow continuously through high-volume ingestions with diverse DML paths, including continuous streams and bulk loading \cite{DBLP:conf/bigdataconf/RuccoLS24}.
Selecting an effective clustering key with the most performance benefit (i.e., a single column or a combination of columns using methods such as Z-ordering \cite{morton1966computer}), poses additional challenges. As query workload patterns evolve, a clustering key that was once optimal may become sub-optimal or even ill-suited, leading to wasted resources and negative performance impacts.

\subsection{Related Work}
\label{sec:related}

\begin{table}[h!]
    \small
    \caption{Flexibility Dimensions of Related Work.
        \textmd{Representative systems classified by their policy and clustering key selection.
            Each cell denotes a distinct clustering maintenance strategy.}}
    \label{tab:flex}
    \begin{tabularx}{\linewidth}{>{\centering\arraybackslash}c
        |>{\centering\arraybackslash}c
        |>{\centering\arraybackslash}X
        |>{\centering\arraybackslash}X}
        \diagbox{\textbf{Policy}}{\textbf{Key}} & \makecell{Fixed}          & \makecell{Dynamic (\textit{Manual})}                           & \makecell{Dynamic (\textit{Workload})}                                              \\
        \hline
        Full Table                              & \ \ legacies\ \           & \makecell{\vspace{1.6em}}                                      & \makecell{Qd-tree \cite{DBLP:conf/sigmod/YangCWGLMLKA20}, Redshift \cite{redshift}} \\
        \hline
        \makecell{New Data}                     & \makecell{\vspace{1.6em}} & \makecell{Iceberg \cite{iceberg}, Delta Lake \cite{deltalake}} & Databricks \cite{databricks}                                                        \\
        \hline
        \makecell{\marginnote{\ul{R3.M2}}Incremental                                                                                                                                                                               \\\review{c}{(\textit{Data-Driven})}}    &  \makecell{\vspace{2em}}      & \makecell{Snowflake \cite{Snowflake}, Dremio \cite{dremio}} &   \\
        \hline
        \makecell{Incremental                                                                                                                                                                                                      \\\review{c}{(\textit{Workload-Aware})}} & \makecell{\vspace{2em}}& & \textbf{WAIR} (ours) \\
    \end{tabularx}
\end{table}

\marginnote{\ul{R4.O4}}[-1.2em]
\review{d}{
    The term ``partitioning'' is overloaded in database literature.
    A significant body of work focuses on \textit{horizontal partitioning} (or sharding) for distributed databases (e.g., Azure Synapse SQL \cite{DBLP:journals/pvldb/Aguilar-Saborit20,azure2025distributionadvisor}) or load balancing in streaming systems (e.g., Dalton \cite{DBLP:journals/pvldb/ZapridouMA22}).
    The primary goal of these techniques is to minimize cross-node traffic and distribute compute load.
    In contrast, modern cloud data warehouses adopt micro-partitioning within the cloud storage layer for \textit{data skipping}.
    This relies on lightweight metadata (e.g., zone maps) \cite{DBLP:journals/pvldb/ZiauddinWKLPK17,DBLP:conf/sigmod/ZimmererDKWOK25} and space-filling curves (e.g., Z-ordering \cite{morton1966computer}) to prune irrelevant micro-partitions during execution.
    In this paper, we target this micro-partition level optimization to facilitate effective pruning, which is orthogonal to and can coexist with prior works.}

A variety of research prototypes and commercial systems\footnote{Because most industrial implementations are proprietary features of commercial cloud data warehouses, our analysis relies on publicly accessible documentation, patents, and published literature gathered on our best-effort basis.} have been proposed to maintain table clustering.
A complete clustering solution usually comprises two algorithmic components:
\textbf{Reclustering Policy} (when/what to reorganize) and \textbf{Clustering Key Selection} (sort order).
We classify representative systems along these two dimensions in \cref{tab:flex} and discuss their limitations.

For the \textit{reclustering policy}, we consider four levels:
1) \textbf{Full Table:} Reorganizing the entire table during maintenance windows;
2) \textbf{New Data:} Clustering only newly ingested data;
3) \textbf{Incremental, Data-Driven:} Reclustering targeted subsets based on data distribution metrics (e.g., overlap) to capture most of the pruning benefit;
4) \textbf{Incremental, Workload-Aware:} The most adaptive class integrates workload statistics, data distribution, and a cloud cost model to prioritize the highest-impact partitions and schedule incremental reclustering at the optimal timing.
For \textit{clustering key selection}, we distinguish between:
1) \textbf{Fixed} keys defined at creation;
2) \textbf{Dynamic, Manual} keys updated by users for further optimization;
3) \textbf{Dynamic, Workload} keys automatically refined by the system, potentially using composite keys and hybrid layouts.

\textbf{Qd-tree}
tailors the block assignment strategy for a given query workload to reduce the number of blocks accessed when running that workload \cite{DBLP:conf/sigmod/YangCWGLMLKA20}.
\textbf{Redshift}
uses a multi-dimensional data layout sort key to reorganize tables at system idles \cite{DBLP:conf/sigmod/DingABPJPPPPPSS24, rsvt}.
While effective for static data, full table repartitioning is computationally prohibitive and less robust for large tables in cloud data warehouses with continuous ingestion and evolving workloads.

\begin{figure}[t!]
    \centering
    \includegraphics[width=\linewidth]{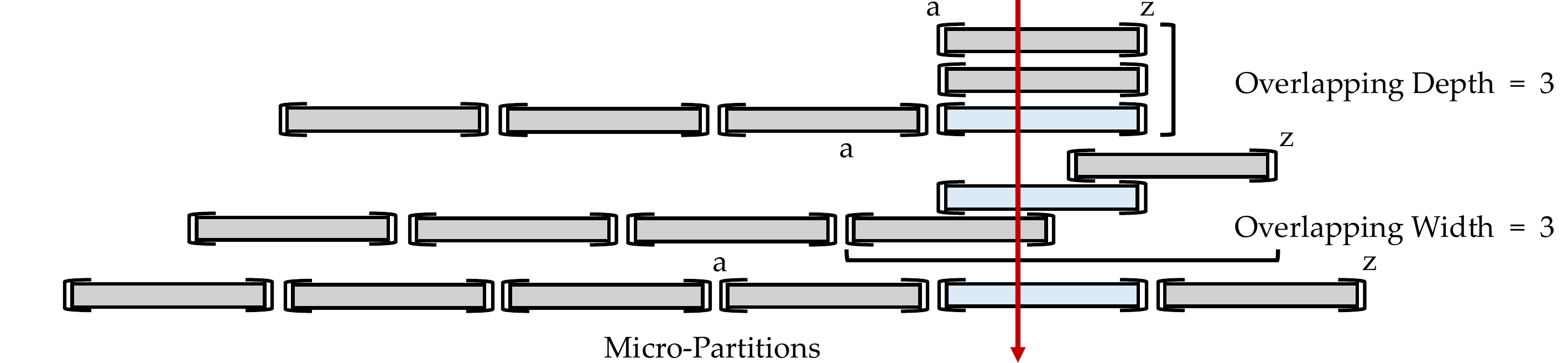}
    \caption{Overlapping Depth and Width. \textmd{Number of micro-partitions overlapping a given micro-partition (width) and a given point (depth). By reclustering, both the overlapping width and depth decrease from three to one.}}
    \Description{Overlapping Depth and Width.}
    \label{fig:overlap}
\end{figure}

\textbf{Iceberg}
supports in-place table evolution, where newly ingested data follow the updated partition layout or sort order \cite{ibe}.
\textbf{Delta Lake}
organizes files into stable ``$ZCubes$,'' each produced by a single \texttt{OPTIMIZE} job. During every optimization cycle, it clusters newly ingested files together with any undersized, partial $ZCubes$ left behind by \texttt{DELETE} operations \cite{dlfrlc, dlddlc}.
\textbf{Databricks}
extends Delta Lake by analyzing query workloads to identify the most effective clustering columns, then applies those keys to subsequent ingests \cite{dblc, dbalc}.
These approaches have two main drawbacks:
1) Existing historical data remains untouched,
and 2) independent ``sorted runs'' across batches degrade global clustering quality over time.

\textbf{Snowflake} evaluates a table's clustering quality by the overlapping metrics \cite{ckct}.
Using partition metadata, it counts the number of micro-partitions in the table whose key ranges overlap with a given micro-partition and a given point (see \cref{fig:overlap}). These metrics averaged across the table reflect its overall clustering quality.
\textbf{Dremio} uses the same metric, where only a small subset of files with the highest overlapping depth within their key ranges are reclustered at optimization cycles, until the table satisfies a user-defined overlapping depth threshold \cite{dmaic}.
While more efficient than full repartitioning, these data-driven techniques:
1) involve potential expensive scans to compute metrics,
2) rely on manual parameter tuning;
3) ignore query patterns, resulting in wasteful reclustering of rarely accessed historical data.

The two strategies outlined above typically require a predefined clustering key, which leads to additional drawbacks:
1) Key selection requires domain expertise.
2) Key updates fragment the layout (old data retains old keys);
3) Even composite keys (e.g., Z-ordering) may fail to serve diverse access patterns effectively.

\section{Boundary Micro-Partitions}
\label{sec:boundary}

A key observation behind incremental reclustering is that most of the query performance benefits come from reclustering only a strategic subset of partitions.
In this section, we introduce the concept of boundary partitions and explain why they are critical for pruning efficiency.
We then propose a simple greedy reclustering algorithm based on boundary partitions with a theoretical analysis of its amortized cost.

\subsection{The Reclustering Workflow}

\begin{figure}[t!]
  \centering
  \includegraphics[width=\linewidth]{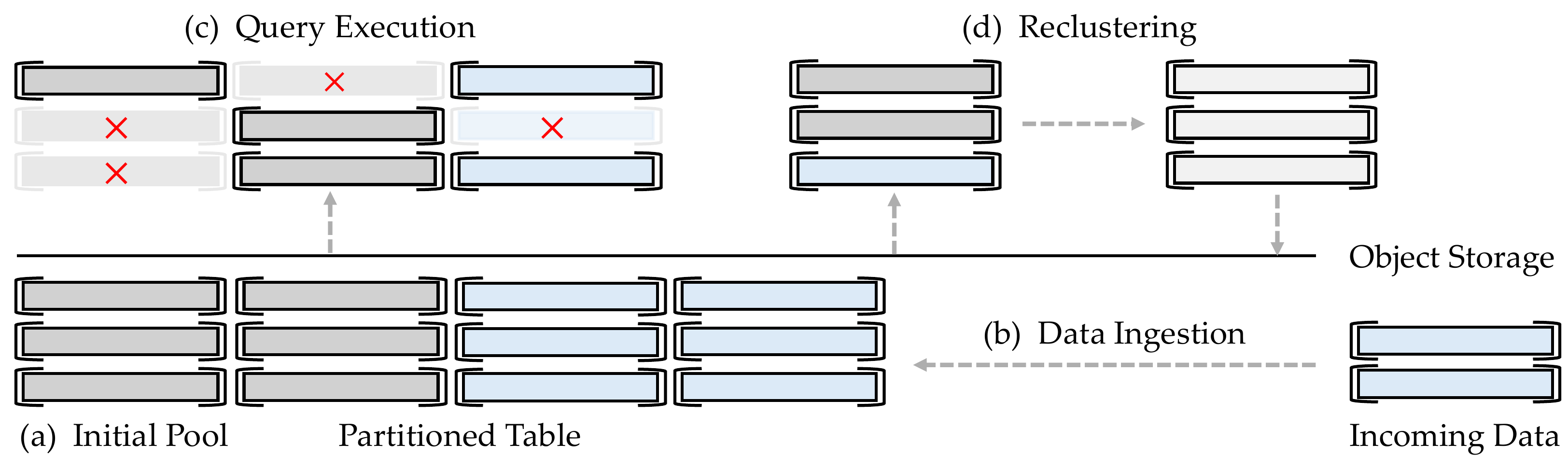}
  \caption{System Model with Reclustering. \textmd{(a) An existing data pool containing previously stored data. (b) Incoming data is partitioned and ingested in its natural order. (c) Query workloads use metadata to fetch only relevant partitions for processing. (d) Partitions are fetched, sorted, and then persisted back into storage.}}
  \Description{System Model}
  \label{fig:system-model}
\end{figure}

\cref{fig:system-model} shows the typical workflow in a cloud data warehouse.
An initial pool of micro-partitions exists in cloud object storage and is clustered based on certain clustering keys.
A \textbf{clustering key} is the column(s) used to sort a micro-partition.
For the theoretical analysis in this section, we assume a clustering key contains a single column.
New data is periodically ingested into database tables as micro-partitions.
These incoming micro-partitions are clustered in their \textbf{natural ingestion order}, which is typically different from the clustering-key order.

Each \textbf{query} in our problem context defines its scan set by specifying the target table along with the pushed-down range predicates.
Meanwhile, \textbf{reclustering} is triggered in the background, following a three-step process:
1) retrieve relevant micro-partitions from object storage,
2) sort the records by the clustering key and assemble new micro-partitions,
and 3) write newly created micro-partitions back to object storage.
Although reclustering is not on the critical path to affect query latency, it consumes additional compute resources.
Therefore, there is a trade-off between the immediate resource expenses on reclustering operations and the long-term resource savings on future queries.

\subsection{Micro-Partitions on the Query Boundaries}
\label{sec:greedy}

\begin{figure}[t!]
  \centering
  \includegraphics[width=\linewidth]{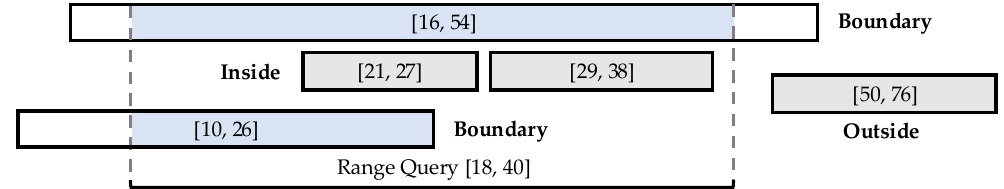}
  \caption{Partitions on the Boundary. \textmd{Boundary partitions overlap with query range edges and primarily affect pruning efficiency.}}
  \Description{Partitions on the Boundary}
  \label{fig:boundary}
\end{figure}

A basic assumption of workload-aware algorithms is that future queries follow a similar statistical pattern as the historical ones.
As depicted in \cref{fig:boundary}, we define \textbf{boundary micro-partitions} as those that \textit{partially} overlap with at least one already-queried range.
Each range query generates two sets of boundary partitions that overlap either the left or right border of its range predicate.
Micro-partitions that fall completely inside or outside any query ranges do not require reclustering because all the records in them are either included or excluded from query results.
Boundary micro-partitions, on the other hand, are critical to pruning efficiency because scanning these micro-partitions leads to retrieving unnecessary rows and therefore wastes I/O bandwidth during query execution.
\cref{alg:greedy} presents a basic greedy algorithm that reclusters all boundary micro-partitions after each query.
We then provide a theoretical analysis to show that this algorithm incurs at most \textit{logarithmic} additional cost.

\begin{algorithm}[t!]
  \caption{A Greedy Algorithm on Boundary Partitions}\label{alg:greedy}
  \begin{algorithmic}
    \Require Partitions $\{P_i\}$ with ranges $\{[a_i,\ b_i]\}$; Predicate $[l,\ r]$
    \vspace{0.5em}

    \State $boundaries \gets \{P_i \mid l \in [a_i,\ b_i]\}$
    \State \Call{Recluster}{$boundaries$}
    \vspace{0.5em}

    \State $boundaries \gets \{P_i \mid r \in [a_i,\ b_i]\}$
    \State \Call{Recluster\footnotemark}{$boundaries$}
  \end{algorithmic}
\end{algorithm}

\footnotetext{Duplicate partitions should be omitted from reclustering; however, their presence does not affect the theoretical analysis of the algorithm.}

\subsubsection*{\textbf{Assumptions}}
Without loss of generality, we assume each value in the clustering-key column is a real number, and each micro-partition is a multiset of $m$ such values.
We represent each micro-partition as $P[a_i, b_i]$ using its clustering-key range.
Consider performing a sequence of \emph{ingestion} and \emph{reclustering} operations on a table with a collection of initial micro-partitions, where

\begin{enumerate}[label=-, leftmargin=*]
  \item \texttt{Ingestion}: Add a new micro-partition of size $m$ into the table. The cost of this operation is excluded from the analysis.
  \item \texttt{Recluster}($x$): Given a value $x$, recluster all the micro-partitions whose range covers $x$ (i.e., $x \in [a_i, b_i]$). Skip micro-partitions with a degenerate range $P[x, x]$. Assume that the cost of reclustering a micro-partition is 1.
\end{enumerate}

Let $n$ denote the total number of micro-partitions ingested.
\marginnote{\ul{R4.O2}}[-1.2em]
\review{d}{
  Let $q$ denote the total number of range queries performed in the entire process.
  The number of reclustering operations is $2q$ because each range query $[{l,\ r}]$ is transformed into two reclustering operations at $l$ and $r$.}
We assume that only keys $1, 2, \dots, \review{d}{2q}$ might be operated with reclustering and all clustering keys lie in the range $(0, \review{d}{2q}+1)$; otherwise, we could rescale them using a piecewise linear function to make this true while preserving their relative orders.

\subsubsection*{\textbf{Potential Function}}
We use the classic potential method~\cite{cormen2022introduction} to examine the \emph{amortized cost} of the operations.
A potential function, typically denoted as $\Phi(S)$, maps each state of the data structure to a non-negative number.
``Potential'' here refers to a conceptual ``energy'' or ``prepaid work'' stored within a data structure.
For a micro-partition $P[a_i, b_i]$, we define its (reclustering) potential as:
\[
  \phi(a_i, b_i) \coloneq
  \begin{cases}
    0                                                     & \text{if $a_i = b_i$ or $\ceil{a_i} > \floor{b_i}$}, \\
    \review{d}{4 + 4 \log (\floor{b_i} - \ceil{a_i} + 1)} & \text{otherwise}.
  \end{cases}
\]

Note that $0 \le \phi(a_i, b_i) \le O(\log q)$.
A micro-partition with a wider key range has a larger $\phi(a_i, b_i)$ because it is more likely to be involved in a future reclustering operation.
The potential of all micro-partitions of a table (i.e., the data structure) is defined as $\Phi = \sum \phi(a_i, b_i)$.
Each ingestion operation increases the potential by $O(\log q)$.

\subsubsection*{\textbf{Micro-Partition Matching}}
Suppose \texttt{Recluster}($x$) takes $k$ micro-partitions $P[a_i, b_i]$, $i = 0, 1, \dots, k-1$ as input and outputs $k$ micro-partitions $P[c_i, d_i]$.
In this case, the potential of each input micro-partition $P[a_i, b_i]$ is at least 1.
We also note that the key ranges of the output micro-partitions are pairwise disjoint except at their endpoints.
Therefore, any output clustering key can exist in at most two non-degenerate micro-partitions.

\begin{lemma}
  We say an output $P[c_j, d_j]$ is \textbf{matched} to an input $P[a_i, b_i]$ if $[c_j, d_j] \subseteq [a_i, b_i]$.
  Then there exists a matching of size $k - 3$ between the input and output micro-partitions (i.e., there are at most 3 unmatched output micro-partitions).
\end{lemma}

\begin{proof}
  Construct a bipartite graph $G$ whose left vertices represent the $k$ input micro-partitions and right vertices represent the $k$ output micro-partitions. Draw an edge from input $P[a_i, b_i]$ to output $P[c_j, d_j]$ if and only if $[c_j, d_j] \subseteq [a_i, b_i]$. We apply Hall's theorem \cite{Hall†1987} to show that $G$ has a matching of size $k - 3$. Let $S$ be a subset of the input micro-partitions, and let $N(S)$ be the neighbors of $S$ in the output micro-partitions. We only need to show that $|N(S)| \ge |S| - 3$ for all $S$.

  Let $A = \min_{i \in S} a_i$ and $B = \max_{i \in S} b_i$. Since each input micro-partition covers $x$, there is an input $P[a_i, b_i] \supseteq [A, x]$ and an input $P[a_j, b_j] \supseteq [x, B]$. Moreover, the input micro-partitions in $S$ provide a total of $m \cdot |S|$ keys in the range $[A, B]$, which implies that there are at least $|S|$ output micro-partitions whose ranges intersect with $[A, B]$. Among these outputs, at most one can have its left endpoint smaller than $A$, and at most one can have its right endpoint greater than $B$; the remaining $|S| - 2$ output micro-partitions must be subintervals of $[A, B]$. Furthermore, at most one output micro-partition $P[c_i, d_i]$ can strictly contain $x$ (i.e., $c_i < x < d_i$), and the remaining $|S| - 3$ outputs must be subintervals of either $[A, x]$ or $[x, B]$, which implies that they are the neighbors of $S$. Therefore, $|N(S)| \ge |S| - 3$ for all $S$, and the lemma follows.
\end{proof}

\begin{lemma}
  For all but $O(\log q)$ matched input-output pairs, $\review{d}{\phi(c_j, d_j) \le \phi(a_i, b_i) - 4}$.
\end{lemma}
\begin{proof}
  We first prove this claim for output micro-partitions $P[c_j, d_j]$ with $c_j \ge x + 12$. Assume there are $r \le k$ such outputs $P[c_1, d_1], \ldots, P[c_r, d_r]$ satisfying $d_i \le c_{i+1}$ for $1 \le i < r$. Let $d_0 = x + 12$ for convenience, and let $P[a_i, b_i]$ be the input micro-partition matched with $P[c_i, d_i]$. Each of these $r$ outputs falls into either case:
  \begin{itemize}
    \item $(d_i - x) \ge 3/2 \cdot (d_{i-1} - x)$. Since the maximum $\floor{d_i}$ is at most $\review{d}{2q}$, there can only be $O(\log q)$ such output micro-partitions.
    \item $(d_i - x) < 3/2 \cdot (d_{i-1} - x)$. This simplifies to $3(d_i - d_{i-1}) < d_i - x$, which further implies
          $$
            d_i - c_i \le d_i - d_{i - 1} \le \frac{d_i - x}{3} \le \frac{d_i - x}{2} - 2 < \frac{b_i - a_i - 1}{2} - 1,
          $$
          where the second-to-last inequality holds because $d_i - x \ge 12$; the last inequality holds because $P[a_i, b_i]$ covers both $x$ and $P[c_i, d_i]$. Therefore,
          \review{d}{$$
              \phi(c_i, d_i) \le 4 + 4 \log (d_i - c_i + 1)\le 4 + 4 \log (\frac{b_i - a_i - 1}{2}) \le \phi(a_i, b_i) - 4.
            $$}
  \end{itemize}
  That is, for all but $O(\log q)$ output micro-partitions $P[c_i, d_i]$ with $c_i \ge x + 12$, we have \review{d}{$\phi(c_i, d_i) \le \phi(a_i, b_i) - 4$.} By symmetry, the same holds for output micro-partitions $P[c_i, d_i]$ with $d_i \le x - 12$. For the remaining partitions, either:
  \begin{itemize}
    \item $P[c_i, d_i]$ is non-degenerate and covers an integer in $[x - 12, x + 12]$. There are only $O(1)$ such micro-partitions, since at most two non-degenerate outputs cover any given integer.
    \item $P[c_i, d_i]$ is degenerate or does not cover any integer. In either case, its potential is 0 and is \review{d}{at least 4 less than its matched input partition (which is at least 4).}
  \end{itemize}
  Putting all cases together completes the proof.
\end{proof}

\review{d}{
  \begin{lemma}
    \label{lem:recluster-cost}
    A reclustering operation of $\,k$ micro-partitions decreases the potential $\Phi$ by at least $4k - O(\log q)$.
  \end{lemma}
}
\begin{proof}
  We analyze the potential change $\Delta \Phi$ of a reclustering operation as follows:
  \begin{itemize}
    \item At most 3 unmatched output micro-partitions are created, which increase the potential by at most $O(\log q)$.
    \item At most 3 unmatched input micro-partitions are removed, which do not increase the potential.
    \item For each matched pair of an input $P[a_i, b_i]$ and an output $P[c_j, d_j]$, we have $\phi(c_j, d_j) \le \phi(a_i, b_i)$.
          For all but $O(\log q)$ matched pairs, we have \review{d}{$\phi(c_j, d_j) \le \phi(a_i, b_i) - 4$.}
  \end{itemize}
  Taking a summation over all these changes yields the lemma.
\end{proof}

\marginnote{\ul{R4.O2}}[-1.2em]
\review{d}{
  If a range query accesses $k$ boundary micro-partitions, it incurs $k$ cost to fetch and another $k$ cost to recluster them.
  Combined with a potential change of at most $O(\log q) - 4k$, the amortized cost of fetching and reclustering boundary micro-partitions is at most $O(\log q)$ per query.
  Non-boundary micro-partitions accessed by a query are fully utilized and their cost is bounded by the query's optimal output size.
  Putting all pieces together, we have:
}

\begin{theorem}
  \label{thm:greedy}
  For a sequence of batched ingestions and range queries, the algorithm achieves a total cost of less than
  \[
    O((n + q) \log q) + \sum \ceil[\big]{|\textup{output}_i| / m},
  \]
  where $\textup{output}_i$ is the number of keys returned by the $i$-th range query.
\end{theorem}

\review{d}{
  \subsection{Greedy Algorithm Extensions}
  \label{sec:adjustments}
  \normalmarginpar\marginnote{\color{black}\ul{R4.O1}}[-1.2em]

  \newcommand{\Cstart}{C_{\mathrm{start}}}
  \newcommand{\kMin}{k_{\mathrm{min}}}
  \newcommand{\kMax}{k_{\mathrm{max}}}
  \newcommand{\tResel}{t_{\mathrm{resel}}}
  \newcommand{\KSel}{K_{\mathrm{sel}}}
  \newcommand{\KNon}{K_{\mathrm{non}}}

  To bridge the gap between theory and practice, we propose three modifications to the greedy algorithm in \cref{alg:greedy}.
  These adjustments address system constraints (e.g., memory limits) while still being amenable to theoretical analysis:

  \begin{enumerate}[leftmargin=*]

    \item \textbf{Memory Limit:} Let $\kMax = \Omega(\log n)$ be the memory capacity for reclustering micro-partitions.
          If a target set $|S| > \kMax$, we split $S$ into disjoint subsets of size $\kMax$ and recluster them separately to avoid expensive external-memory sorting.

    \item \textbf{Warm Start:} Let $\Cstart = \Omega(\log n)$ be a chosen parameter.
          We enforce a budget of $\Cstart \cdot t$ over the baseline cost during the first $t$ steps to prevent initial cost spikes.
          A subset of proper size is allowed if a reclustering operation violates this constraint.

    \item \textbf{Clustering Key Reselection:} Every $\tResel = \Theta(n)$ operations, we re-evaluate the clustering key for subsequent reclustering operations to use.
          This enables adaptation to workload shifts while preventing rapid oscillations.
  \end{enumerate}

  We analyze the cost of the extensions using the same potential function analysis as in \cref{sec:greedy}.
  By \cref{lem:recluster-cost}, a reclustering operation of $\kMax$ micro-partitions incurs non-positive amortized cost.
  Thus, each range query and its subsequent reclustering operations incur an amortized cost of at most $O(\log n)$ under the memory limit.
  The warm-start phase affects only the initial $O(n)$ steps and does not impact the amortized cost thereafter.
  Periodic clustering-key reselection adds $O(n \log n)$ potential every $\Theta(n)$ operations, incurring at most $O(\log n)$ to the amortized cost per operation.
  Putting all pieces together, we have the following theorem.

  \begin{theorem}
    For a sequence of batched ingestions and range queries, the adjusted greedy algorithm achieves a total cost of at most
    \[
      O((n + q) \log n) + \sum \ceil[\big]{|\textup{output}_i| / m} + \textup{(warm starting cost)},
    \]
    where $\textup{output}_i$ is the number of keys returned by the $i$-th range query, and the warm starting phase lasts for at most $O(n)$ operations.
  \end{theorem}
}

\section{WORKLOAD-AWARE RECLUSTERING}
\label{sec:wair}

In this section, we build upon the concepts of boundary micro-partitions and propose a practical automatic reclustering framework.
We first introduce a key metric ``utilization'' to quantify the expected savings of reclustering each boundary micro-partition.
We then present a cost model that trades off between the reclustering overhead and the expected query execution savings to guide the proper timing and micro-partition selection for reclustering.
Finally, we propose a hybrid layout where different micro-partitions can adopt different clustering keys to maximize the reclustering benefit under a diverse workload.

\subsection{Micro-Partition Utilization}

\begin{figure}[t!]
    \centering
    \includegraphics[width=\linewidth]{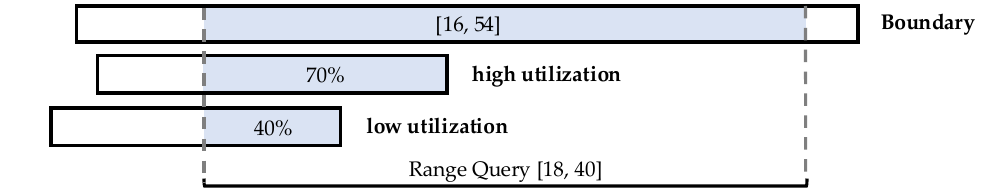}
    \vspace{-1.5em}
    \caption{Micro-Partition Utilization. \textmd{Utilization measures percentage of a micro-partition's bytes actually needed by the query.}}
    \Description{Micro-Partition Utilization}
    \label{fig:utilization}
\end{figure}

As illustrated in \cref{fig:utilization}, we define the \textbf{utilization} of a micro-partition with respect to a query as the ratio between the number of bytes read by the query and the total number of bytes of the micro-partition\footnote{The utilization calculation in our prototype system considers practical aspects such as projection pushdown, row groups in storage format, and data encoding and compression (see \cref{sec:impl}).}.
A low utilization indicates poor clustering: the system reads the full micro-partition, yet the query uses only a small fraction of the retrieved rows.
Reclustering poorly utilized micro-partitions can yield more pruning benefits.
For example, reclustering 100 micro-partitions at utilization $10\%$ could generate 10 fully utilized micro-partitions for the same query, thus reducing the scan set by $\approx10\times$ for similar future queries.

\subsection{Cost-based Trade-off}
\label{sec:cost}

In cloud data warehouses, reclustering imposes a one-time cost but delivers future performance benefits.
An effective reclustering policy, therefore, initiates the operation only when its expected gains outweigh incurred costs.
To quantify this trade-off, we develop a cost model that translates each micro-partition's utilization into predicted query-time savings.
By comparing predicted savings with reclustering overhead, the cost model guides our reclustering policy to the optimal timing and target micro-partition set.

Let $cost_{q}(Q)$ denote the cost of executing a single query $Q$.
Let $cost_{r}(\mathcal{P})$ represent the cost of reclustering a set of micro-partitions $\mathcal{P} = \{P_i\}$.
Given a preceding sliding time window $W$, as shown in \cref{fig:window}, we estimate the cost reduction $\Delta \hat{cost}_{q}(\mathcal{P} \mid Q)$ if the micro-partition set $\mathcal{P}$ were reclustered.
Then, the \texttt{Recluster}$(\mathcal{P})$ operation is triggered only when the predicted aggregate savings from all queries in $W$ outweigh the estimated reclustering expense, and the current reclustering ``debt'' is under a predefined limit:
\begin{equation}
    \Delta \hat{cost} = \hat{cost}_{r}(\mathcal{P}) - \sum_{Q\in W} \Delta \hat{cost}_{q}(\mathcal{P} \mid Q) < 0
\end{equation}
\begin{equation}
    \label{eq:cost-limit}
    \Delta cost = \sum_{\mathcal{P}_i \in W}cost_r(\mathcal{P}_i) - \sum_{Q \in W} \Delta cost_{q}(Q) < cost\_limit
\end{equation}

\begin{figure}[t!]
    \centering
    \includegraphics[width=\linewidth]{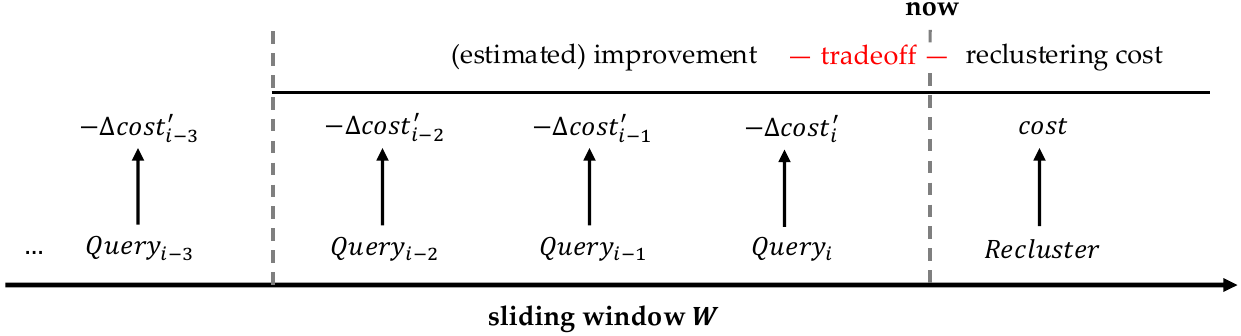}
    \vspace{-1.5em}
    \caption{Tradeoff between Reclustering Cost and Estimated Improvement. \textmd{Reclustering incurs a fixed overhead. Its potential benefit is estimated by summing the projected improvements over a preceding sliding time window. If the estimated benefit outweighs its overhead, a reclustering operation is triggered.}}
    \Description{Tradeoff between Reclustering Cost and Estimated Improvement.}
    \label{fig:window}
\end{figure}

The purpose of the cost limit is to bound the wasted effort under a dramatic workload shift.
We next express the above estimated costs in CPU times.
For a query $Q$, the estimated execution time reduction contributed by a particular micro-partition $P$ is:
\begin{equation}
    \label{eq:query-time}
    \Delta \hat{t}_{q}(P\mid Q) = \frac{(1 - u) \cdot size_{read}(P)}{size_{read}(Q)} \cdot t_{read}(Q)
\end{equation}
where $u$ denotes the utilization of $P$ with respect to $Q$, and $(1 - u) \cdot size_{read}(P)$ represents the scan size that could be saved after reclustering $P$.
Then $(1 - u) \cdot size_{read}(P) / size_{read}(Q)$ represents the estimated fraction of query read time saving.

The time estimation $t_r(\mathcal{P})$ for reclustering a micro-partition set $\mathcal{P}$ depends on whether $\mathcal{P}$ fits in memory.
If they fit, we perform an in-memory quicksort on the records according to their clustering keys.
Otherwise, we fall back to an external merge sort.
In both cases, $t_r(\mathcal{P})$ is expressed as a function of the aggregate size of $\mathcal{P}$.

\begin{algorithm}[t!]
    \caption{Workload-Aware Reclustering Policy}
    \label{alg:cap}
    \begin{algorithmic}
        \Require All micro-partitions $\mathcal{P}$; Time window $W$
        \For{$P \in \mathcal{P}$}
        \State $\Delta \hat{cost}_q(P) \gets \sum_{Q \in W} \Delta \hat{cost}_q(P\mid Q)$
        \EndFor

        \vspace{0.5em}
        \State $[P_i] = sort_{desc}(\{\Delta \hat{cost}_q(P), P \in \mathcal{P}\})$
        \For {$cut \in \{1, 2, \dots, \lvert \mathcal{P}\rvert\}$}
        \State $[P_i]_{cut} = \{P_1, P_2, \dots, P_{cut}\}$
        \State $\Delta \hat{cost} = \hat{cost}_{r}([P_i]_{cut}) - \sum_{P \in [P_i]_{cut}} \Delta \hat{cost}_{q}(P)$
        \EndFor
        \State $cut^* \gets \mathop{\arg\min}_{cut} \Delta \hat{cost}$ \Comment the smaller the better

        \vspace{0.5em}
        \If{$\Delta \hat{cost} < 0$ \textbf{and} \\\hspace{0.8em} current $\Delta cost + \hat{cost}_{r}([P_i]_{cut^*}) < cost\_limit$}
        \State \Call{Recluster}{$[P_i]_{cut^*}$}
        \EndIf
    \end{algorithmic}
\end{algorithm}

Additionally, we dynamically adjust the sliding window size based on its prediction accuracy to adapt to the evolution of the query workload and the data distribution.
After completing a window $W$ of queries, we compare the actual cost reduction $\Delta cost_q = \sum_{Q \in W}\Delta cost_q(Q)$ against the predicted cost reduction $\Delta \hat{cost}_q = \sum_{Q \in W}\Delta \hat{cost}_q(Q)$.
If $\Delta cost_q > \Delta \hat{cost}_q$, it indicates that the workload is relatively stable, and we therefore double the window size to capture a longer query history.
Otherwise, we halve the window size to focus more narrowly on the most recent workload shifts and distribution changes.

\cref{alg:cap} summarizes our workload-aware reclustering policy.
The system triggers potential \texttt{Recluster} operations right after completing each query.
To determine the set of micro-partitions for reclustering, we loop through all the micro-partitions accessed in $W$.
For each micro-partition $P$, we first estimate the cost reduction $\Delta \hat{cost}_q(P)$ for all queries in $W$ that could benefit from reclustering $P$.
We then sort the micro-partitions according to the cost reduction in descending order and determine a $cut$ in the sorted sequence where reclustering all the micro-partitions before the $cut$ yields the largest estimated cost reduction.
Finally, we perform these reclustering operations subject to a cost limit, as shown in \cref{fig:window}.
The size of $W$ is adjusted after each sliding window as described above.
Reclustering decisions primarily reuse query-execution statistics already logged by the system, incurring minimal overhead.

\review{a}{
    \subsection{Cost Model Extensions}
    \label{sec:extensions}
    \marginnote{\color{black}\ul{R1.O4}}[-1.2em]

    We introduce two cost model extensions to allow adjusting reclustering policies to specific budget constraints or performance goals.

    \smallskip
    \noindent\textbf{(1) Reclustering Aggressiveness.}
    Our standard cost model triggers reclustering when estimated query savings exceeds immediate reclustering costs.
    To accommodate users' different risk preferences, we extend the model with an aggressiveness factor $\alpha > 0$ and a dynamic budget constraint:
    \begin{equation}
        \label{eq:alpha}
        \sum_{Q\in W} \Delta \hat{cost}_{q}(\mathcal{P} \mid Q) > \alpha \cdot \hat{cost}_{r}(\mathcal{P})
    \end{equation}
    \begin{equation}
        \label{eq:credits}
        \hat{cost}_{r}(\mathcal{P}) \le \sum_{j > i} \Delta {cost}_{q}(\mathcal{P}_i \mid Q_j) - \sum cost_{r}(\mathcal{P}_i) + c \cdot \vert Q \vert
    \end{equation}
    As shown in \cref{eq:alpha}, reclustering is triggered when the estimated savings are at least $\alpha$ times the reclustering cost.
    The smaller the $\alpha$, the more proactive the policy.
    \cref{eq:credits} enforces a reclustering budget based on realized savings.
    Query savings are measured by comparing zonemaps before and after reclustering to track additional micro-partitions skipped at future queries.
    A reclustering operation is permitted only if its cost does not exceed the accumulated ``credit'' (i.e., net savings from past reclustering decisions plus a per-query ``allowance'' $c$).

    \smallskip
    \noindent\textbf{(2) Interplay with Performance.}
    While our primary objective is to minimize cost, users could prioritize query latency over cost aggressively.
    Integrating performance into a cost-based model is challenging because the monetary value of query speed varies by application.
    We discuss how our cost model can be extended to accommodate performance-critical scenarios.
    Assuming that all queries are equally important, we introduce a parameter $\beta$ to denote the price a user is willing to pay for every additional pruned micro-partition.
    Then, for each query $Q \in W$, we augment its estimated savings with a performance bonus:
    \begin{equation}
        \Delta \hat{cost}_{q}(\mathcal{P} \mid Q) \gets \Delta \hat{cost}_{q}(\mathcal{P} \mid Q) + \beta \cdot \vert\Delta \hat{pruned}(\mathcal{P} \mid Q)\vert
    \end{equation}
    where $\vert\Delta \hat{pruned}(\mathcal{P} \mid Q)\vert$ denotes the estimated increase in skipped micro-partitions for $Q$ after reclustering $\mathcal{P}$.
    By tuning $\beta$, users can express their specific exchange rates between cost budget and query performance.
}

\subsection{Hybrid Layout}
\label{sec:hybrid}

\begin{figure}[t!]
    \centering
    \includegraphics[width=\linewidth]{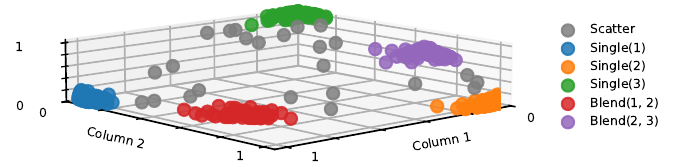}
    \vspace{-1.5em}
    \caption{Hybrid Layout for Clusters. \textmd{After clustering, each partition group is assigned a physical layout that best suits its savings signature, producing a hybrid table layout.}}
    \Description{Hybrid Layout for Clusters}
    \label{fig:hybrid-layout}
\end{figure}

Unlike prior data-driven methods that must commit to a pre-selected clustering key, our approach decouples the reclustering policy from the clustering key selection by providing the flexibility for each micro-partition to choose its most beneficial clustering key.

To understand how each micro-partition might benefit from reclustering, we decompose its aggregate estimated savings into per-column components.
We assume that each predicated column contributes equally in a query.
Specifically, for each selected micro-partition $P$, we distribute its aggregated savings across all columns that appear in a query's predicate within the sliding window:
$$ s_P=\bigl(s_P^{(1)}, s_P^{(2)}, \dots , s_P^{(d)}\bigr), $$
where $d$ is the number of attributes, $s_P^{(i)}$ is the share of $P$'s projected benefit attributed to column $i$.
Each savings vector is then plotted in a high-dimensional space to create the micro-partition's ``\textbf{savings signature}'', as illustrated in \cref{fig:hybrid-layout}.

\marginnote{\ul{R4.O1}}[-1.2em]
\review{d}{
    Building on the theoretical analysis in the previous section, a straightforward approach is to derive a single global clustering key corresponding to the base column with the highest aggregated savings.
    While this strategy inherits the theoretical cost guarantees, practical workloads often exhibit heterogeneous access patterns where no single key is optimal for all partitions.
    To address this, we leverage the decoupling of our framework to propose a \textbf{hybrid layout} strategy.
    This approach allows different micro-partitions to adopt different clustering keys, trading the strict theoretical simplicity for greater flexibility and empirical performance gains.}

We first normalize each vector to unit length and then cluster the normalized points with DBSCAN \cite{DBLP:conf/kdd/EsterKSX96} using cosine distance.
To interpret the resulting clusters, we define two families of anchor points: 1) the unit basis vectors $e_i$ (single-column dominance), and 2) uniform multi-column anchors $v_S$ with equal weights (multi-column blend). After clustering, each cluster is labeled with the anchor that is closest to its centroid.

\begin{itemize}
    \item \textbf{Single-column dominance.} If a micro-partition's projected savings are contributed almost entirely by a single column, reclustering using the column as the clustering key is sufficient. Micro-partitions within the same cluster are grouped and sorted together by their shared dominant column.
    \item \textbf{Multi-column blend.} When savings concentrate on a small set of columns, we recluster using a Hilbert-curve ordering \cite{DBLP:journals/sigmod/LawderK01} over those columns. The columns are arranged in descending order of their overall contribution.
    \item \textbf{High-dimensional scatter.} Micro-partitions with diffuse savings across many columns do not benefit from a lexicographic key. Instead, we organize these micro-partitions with Qd-tree \cite{DBLP:conf/sigmod/YangCWGLMLKA20} to adapt to arbitrary query patterns.
\end{itemize}

The candidate micro-partitions are therefore divided into mutually exclusive groups, each assigned the physical layout that promises the greatest benefit.
This strategy produces a hybrid table organization: different regions of the data are clustered with different keys, yet every reclustering decision is still guided by the same workload-aware cost model described earlier.
Micro-partitions thus shift between groups as the workload evolves, allowing the table layout to adapt over time.

\section{RECLUSTERING AS A SERVICE}
\label{sec:ars}

\begin{figure}[t!]
    \centering
    \includegraphics[width=\linewidth]{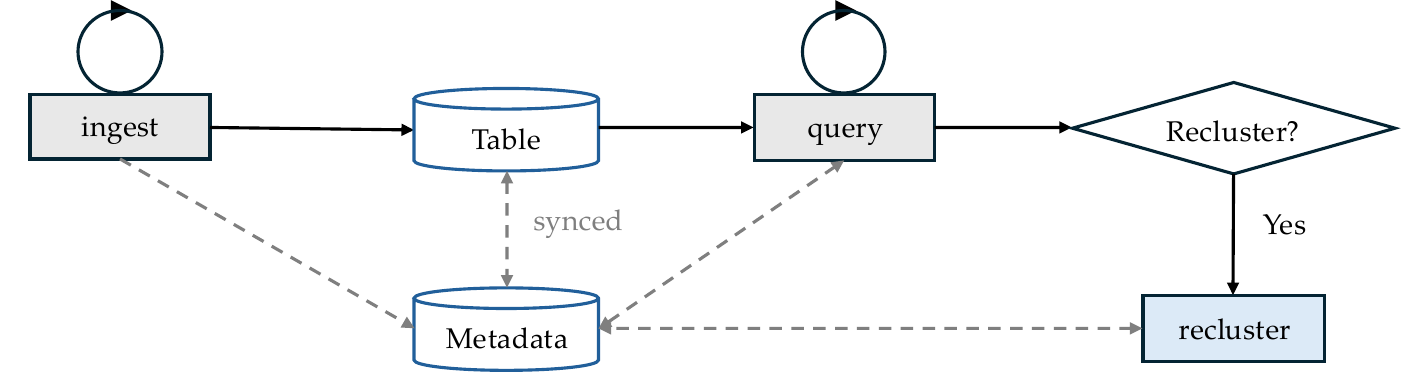}
    \caption{Service Framework Workflow. \textmd{A metadata service continuously collects partition metadata, workload execution statistics, and reclustering snapshots. After a query completes, a reclustering operation may be triggered, which uses the metadata to determine its scope and is then executed on a dynamically provisioned compute node.}}
    \Description{Service Framework Workflow}
    \label{fig:workflow}
\end{figure}

Having presented the \algname algorithm, we now introduce the \sysnamefull (\sysname), a scalable reclustering framework that autonomously manages background operations to preserve table clustering quality over time.
\sysname leverages execution statistics with minimal compute overhead to decide the scope and timing of reclustering.
\sysname moves all reclustering work off the query-execution critical path, ensuring zero disruption to production workloads in practice.
In this section, we describe how to embed the \algname algorithm into \sysname with minimal engineering effort and \marginnote{\ul{R3.O1}}[-1.2em] \review{c}{full compatibility on top of cloud data warehouses.}
We present the practical feasibility of \sysname and the implementation details to be evaluated, and \marginnote{\ul{R1.O1}}[-1.2em] \review{a}{provide an overview of pruning opportunities across representative query patterns.}

\subsection{Service Framework}
\label{sec:distributed}

As shown in \cref{fig:workflow}, integrating the reclustering framework into an existing cloud data warehouse additionally requires: 1) an (often already in place) metadata service, and 2) a dedicated reclustering job executor.

\textbf{Metadata Service.}
Three main categories of metadata are collected: 1) \emph{Partition-level table statistics}, including \reversemarginpar\marginnote{\ul{R1.M2}}[-1.2em] \review{a}{min-max values, sizes, and compression ratios for each column chunk collected from partition metadata}; 2) \emph{Partition utilization statistics}, capturing detailed partition-level information for each query, such as partition pruning outcomes, data access volume, and measurements of filtered rows along with their sizes; 3) \emph{Query-level statistics}, including total I/O size, I/O time, network transfer speed, and CPU time.

Metadata is continuously tracked and preserved during data ingestion and query execution asynchronously for further utilization.
\marginnote{\ul{R3.O1}}[-1.2em]
\review{c}{In modern distributed database systems, the metadata service is typically a standalone component.
    It should integrate seamlessly with a cloud data warehouse's existing metadata infrastructure as a drop-in module.}

\textbf{Reclustering Executor.}
After specifying the partition set to be reclustered and corresponding layouts, the reclustering executor retrieves these partitions from storage layer, loads their data into local memory, and reorganizes them into new partitions.
These clustered partitions are then persisted back to the object storage.

\normalmarginpar\marginnote{\ul{R3.O1}}[-1.2em]
\review{c}{Reclustering operations are designed to be non-blocking in distributed systems.}
When a reclustering operation completes, \sysname publishes a new table snapshot that references the current set of partitions.
Because reclustering reorganizes data without changing any row values, queries that began before completion continue to read the prior snapshot along with any newly ingested partitions.
New queries started after completion read the reclustered partitions referenced by the latest, updated snapshot.
Once no queries reference an older snapshot, \sysname safely removes obsolete partitions during garbage collection.
The table snapshot chain incurs minimal storage overhead and allows new queries to benefit immediately once reclustering completes.
Reclustering tasks run on dynamically provisioned compute nodes, entirely outside the query execution critical path.
\marginnote{\ul{R3.O1}}[-1.2em]
\review{c}{In cloud data warehouses with a cache layer, reclustering naturally benefits from scan sharing as it requires scanning complete micro-partition data.}
Moreover, systems can opt for the k-th quantiles method \cite{cormen2022introduction} over standard, well-optimized sorting to recluster partitions in $O(n\log k)$ time.

\subsection{Implementation Details}
\label{sec:impl}

We implement the service using DuckDB \cite{DBLP:conf/sigmod/RaasveldtM19,DBLP:conf/vldb/Raasveldt22} as the execution engine and deploy in a cloud environment backed by AWS S3 and Redis \cite{redis} (metadata). Data is stored in Parquet \cite{parquet}. During ingestion, we use Arrow \cite{arrow} to write data into fixed-size partitions and extract partition-level statistics from file headers into Redis.

We instrument DuckDB to extract predicates from the ongoing query
and evaluate against the active table snapshot in Redis.
We modify the execution engine to collect relevant statistics when executing queries directly on S3.
Data transfer size of each partition is recorded by patching the \texttt{HTTPFS} extension.
Data chunk before and after filtering are captured at the \texttt{TableScan} executor.

We consider three major optimizations in DuckDB: 1) projection pushdown, 2) column chunks, and 3) encoding and compression.
For unpruned partitions, we record both the total column chunk size retrieved and the actual size read by a query in each column after projection pushdown.
A partition's utilization is then adjusted according to each column chunk's compression ratio.
These cost metrics are evaluated asynchronously to determine a reclustering decision.
After a reclustering job finishes,
the metadata of these new partitions is recorded,
and a new table snapshot is published to keep production workloads uninterrupted.

\review{a}{\subsection{Pruning Opportunities}}
\label{sec:pruning}

\normalmarginpar\marginnote{\ul{R1.O1}}[-1.2em]
\review{a}{
    Zonemap-based pruning extends beyond simple filters to support complex operations.
    We summarize the pruning opportunities for representative query types, as documented by Snowflake \cite{DBLP:conf/sigmod/ZimmererDKWOK25}.
    For \textbf{LIMIT} clauses, a well-clustered table facilitates early termination once enough \textit{fully-matching} partitions are found.
    For \textbf{JOIN} operations, build-side values are summarized into zonemaps for the probe side to prune partitions, which are essentially range queries that benefit from our reclustering approach.
    Snowflake builds a \textit{pruning tree} to handle complex expressions across \textbf{diverse predicates}.
    In our context, each node in this tree identifies a subset of boundary partitions associated with the base columns.
    Our approach improves the efficiency of this pruning tree, and the resulting evaluation process guides our reclustering decisions.
}

\section{EVALUATION}
\label{sec:evaluation}

\marginnote{\ul{Meta}\\\ul{R1.O1}\\\ul{R1.O4}\\\ul{R4.O1}}
\review{m}{We evaluate our designs on \emph{TPC-H} \cite{tpch}, \emph{DSB} \cite{DBLP:journals/pvldb/DingCGN21}, and real-world workload \emph{Mirrors}.}
Our principal findings are:

\begin{enumerate}
    \item \textbf{Incremental Reclustering} is more efficient and practical than full table repartitioning, achieving effective partition pruning at lower reclustering cost.
    \item \textbf{\algnamefull} consistently outperforms baselines, delivering significant reductions in I/O volumes and overall costs.
    \item \textbf{\algname} remains robust \review{m}{across diverse data distributions, continuous ingestion, and evolving workload patterns.}
\end{enumerate}

\review{a}{
    \subsection{Workload Generator}
    \label{sec:workload-generator}

    \marginnote{\color{black}\ul{R1.O1}}[-1.2em]
    Standard benchmarks (e.g., TPC-H) largely rely on static datasets with fixed query templates and limited dynamic features.
    We therefore develop a workload generator to extend standard benchmarks with a set of configurable parameters to capture the interplay between continuous data ingestion and evolving workload patterns.
    The generator divides a workload into consecutive \textbf{periods}.
    Each period is configured to represent a specific workload pattern and consists of a sequence of workload \textbf{batches} (e.g., 12 batches per period).
    Each batch includes a data ingestion phase, a query execution phase, and an optional reclustering phase\footnote{Reclustering is invoked after each batch but may decide to take no action.}.
}

\review{a}{
    \textbf{Data Ingestion.} For static benchmarks such as TPC-H, we slice the dataset into chronological batches (e.g., grouped monthly by order date), preserving the original data ordering.
    For DSB and Mirrors (our real-world dataset), because their datasets have built-in timestamps, we use ``refresh runs'' in DSB and ``replay logs'' in Mirrors to generate the ingestion batches.
    An initial, configurable number of ingestion batches forms the initial state of the database.

    \textbf{Query Mix.} For each batch, the workload generator constructs a \textbf{query mix} derived from the benchmark's query templates.
    A query mix can cover the complete set of query templates or target specific subsets (e.g., join-heavy queries).
    We introduce fine-grained control over queries' \textbf{selectivity} and data access patterns.
    For example, the query predicates can follow a specific \textbf{distribution} (e.g., uniform, skewed towards recent).
    Within each period, we use the \textbf{shifting rate} to control the percentage of queries regenerated between consecutive batches.
    Between periods, the main \textbf{predicated column} may switch to represent a major workload shift (e.g., from ship date to commit date in TPC-H).
}

\subsection{Workloads}
\label{sec:workload}

\marginnote{\ul{R1.O1}}[-1.25em]
\review{a}{We describe how the above workload generator applies to TPC-H, DSB, and Mirrors.}

\subsubsection{TPC-H}

We generate a TPC-H dataset with a scale factor of 720.
We divide the dataset into 72 monthly batches grouped by (\texttt{o\_orderdate})\footnote{The maximum gap between \texttt{l\_shipdate} and \texttt{o\_orderdate} is extended from TPC-H default of 121 days to 1,000 days accordingly to create a larger mismatch between ingestion order and query predicates.}.
Reclustering targets the largest fact table \texttt{lineitem}.
Each batch comprises all 8 query templates that filter on \texttt{l\_shipdate}.
The default predicate selectivity is set to a two-month interval.
Each batch runs: 1) a \texttt{global} workload of the 8 queries with uniformly sampled predicates across the entire date range, and 2) a \texttt{local} workload of the 8 queries following a Zipf ($\alpha=2$) distribution skewed toward recent months.

The workload includes a total of 6 periods ($\mathbf{P1}$ to $\mathbf{P6}$), each consisting of 12 consecutive batches.
The first two periods $\mathbf{P1, P2}$ form the initial data pool.
The workload maintains a shifting rate of 25\% in $\mathbf{P2}$ and $\mathbf{P3}$ ($\mathbf{P2}$ for warm up).
Reclustering begins from $\mathbf{P3}$ and continues thereafter.
Workload shifting rate rises to 75\% from $\mathbf{P4}$.
To mimic a major workload shift, the primary predicate column changes to \texttt{l\_commitdate} in $\mathbf{P5}$, and to a 2:1 mix of \texttt{l\_shipdate} and \texttt{l\_commitdate} in $\mathbf{P6}$.

\marginnote{\ul{Meta}\\\ul{R1.O1}\\\ul{R4.O1}}[0.65em]
\review{m}{
    \subsubsection{DSB}
    \label{sec:dsb}

    Compared to TPC-H, DSB has more complex schema with multiple fact tables.
    DSB's query templates target (i) complex joins (averaging a degree of 10.8), (ii) more filters on diverse table columns, and (iii) LIMIT queries \cite{DBLP:journals/pvldb/DingCGN21}.

    We follow settings similar to the TPC-H setup above.
    The workload consists of 6 periods, each containing 12 batches.
    An initial dataset (SF=240) and these 72 ``refresh runs'' ingestions bring the final scale factor to around 720.
    The first two periods ($\mathbf{P1, P2}$) expand the initial data pool.
    Reclustering begins from $\mathbf{P3}$ and targets the three largest fact tables: \texttt{store\_sales}, \texttt{catalog\_sales}, and \texttt{web\_sales}.
    Each batch includes all 30 DSB queries that filter on these tables.
    Predicate values are sampled from DSB's default Gaussian ($\sigma=2$) distribution, with dates centered on a reference month that advances by one month per batch.
    We set the workload shifting rate to 100\% and fix query selectivity at a two-month interval.
}

\subsubsection{Mirrors}

Our real dataset, called \textbf{Mirrors}, comprises access logs collected from
a major open-source software mirror \cite{tuna}.
Spanning the last 120 days, it contains 2.4TB of raw web server logs, which compress to a total of 212GB.
Each record includes source address, user agent, access path, response status code, response size, and multiple access related information fields\footnote{An obfuscated example: 117.176.220.183 - - [29/Sep/2024:10:00:31 +0800] "GET \url{/p
    ypi/web/packages/f6/ab/c7d5e79d2984001911d864af8ec74492da5dba558737b10774ce27587
    239/duckdb-1.1.0-cp310-cp310-manylinux_2_17_x86_64.manylinux2014_x86_64.whl} HTTP/1.1" 200 20097722 "application/octet-stream" "-" "poetry/1.8.3 CPython/3.10.15 Linux/5.15.153.1-microsoft-standard-WSL2" - https} \cite{tunalogs}.

Massive suspicious IP addresses are identified and blocked by the web firewall.
The dataset is paired with an internal workload derived from user tickets requesting IP unblocking.
To address these requests, we investigate the access logs across a range of associated IP addresses to determine if the flagged addresses are indeed safe.
The workload involves queries for distinct IP counts, user agent counts, access path frequencies, and the distinct count of user agents associated with distinct IP addresses.
These analyses are performed at the \texttt{/16} subnet level for IPv4 addresses and \texttt{/32} for IPv6.
Over the 120-day period, 1,210 unique IP addresses with their corresponding ranges are queried for potential unblocking, with an overall workload selectivity of 0.1\%.
Reclustering begins on day 30, using the initial 30 days of data as a natural pool.

\subsection{Experimental Setup}
\label{sec:setup}

\begin{figure}[t!]
    \centering
    \includegraphics[width=0.95\linewidth]{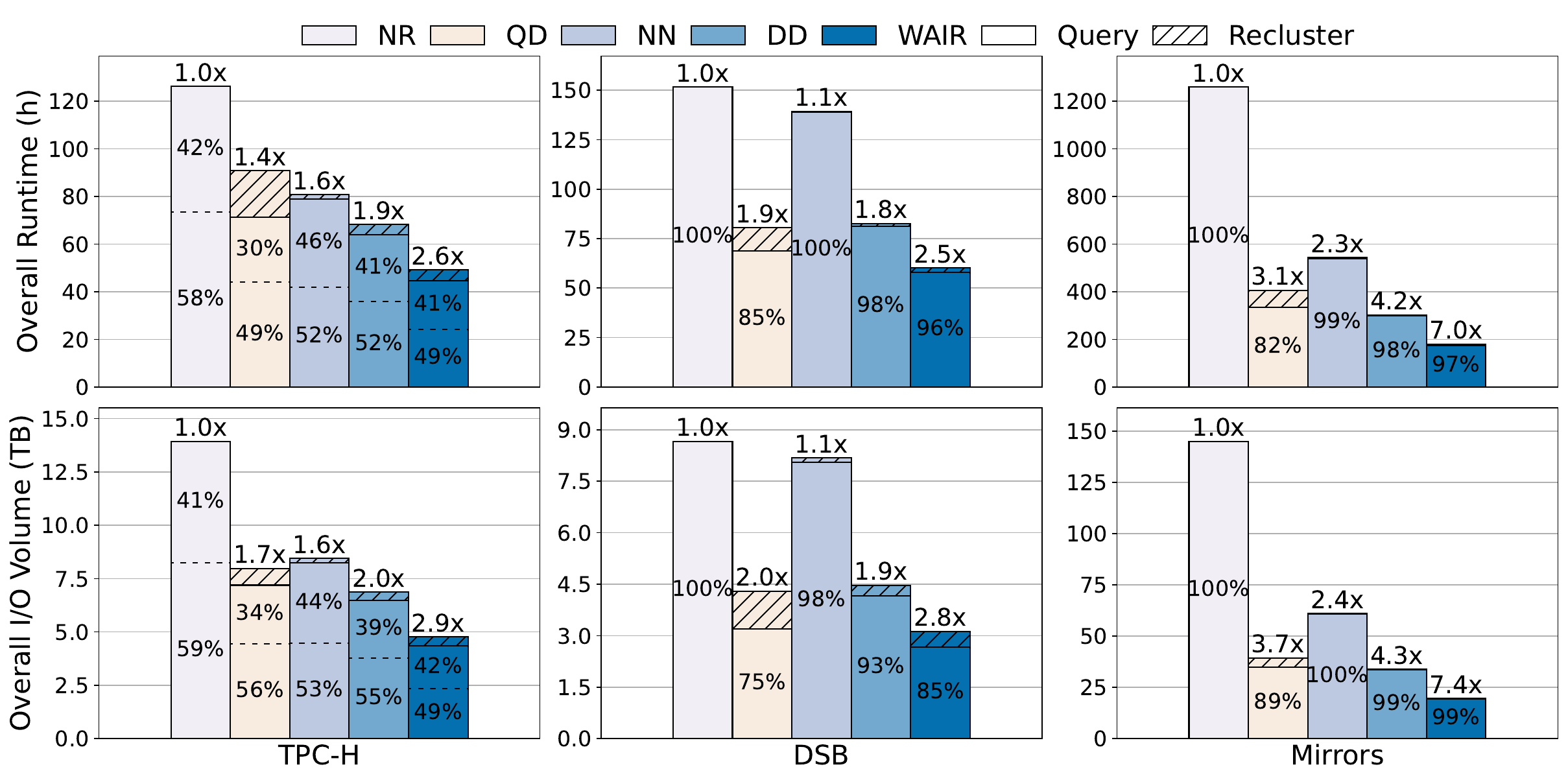}
    \caption{\review{m}{End-to-End Results.}
        \textmd{\review{m}{Each bar shows total runtime or I/O volume per method on TPC-H, DSB and Mirrors; numbers above bars indicate their overall speedup relative to NR.
                Bars are split into query (solid) and reclustering (hatched); for TPC-H, query is further divided into local- (lower) and global- (upper) components.
            }}}
    \marginnote{\ul{Meta}\\\ul{R1.O1}\\\ul{R1.O3}\\\ul{R4.O1}}[-6.3em]
    \Description{Overall Results.}
    \label{fig:overall}
\end{figure}

We compare \textbf{\algname}\footnote{\marginnote{\ul{R1.O4}}\review{a}{with the standard cost model and a hybrid layout as default.}}, our workload-aware approach against four alternate approaches described in \cref{sec:preliminaries}:

\begin{enumerate}
    \item \textbf{No Recluster (NR)}: A simple baseline that keeps the data untouched in its original ingest order without applying any reclustering.
    \item \textbf{Qd-tree (QD)}: Build a Qd-tree from historical access patterns and then fully repartition the table at the start of TPC-H
          \review{m}{and DSB periods,} and the end of months in the Mirrors workload.
    \item \textbf{New Data New Cluster (NN)}: A Delta Lake-style approach that sorts newly ingested data in each batch into stable \textit{ZCubes}, leaving existing data unchanged.
    \item \textbf{Data-Driven Incremental (DD)}: We choose Dremio to represent the data-driven incremental approach. Dremio uses Snowflake's overlapping metric to rank partitions and, in each batch, reclusters up to a fixed number of partitions whose overlapping depth exceeds a predefined overlapping threshold.
\end{enumerate}

All data are stored as 32MB Parquet files\footnote{We repeated with various partition sizes and observed consistent results.} with a single row group per file.
To ensure a fair comparison, we evaluate all approaches under the same cloud configuration and random seed.
Once a method selects partitions and the sort order for reclustering, the identical reclustering code is applied.
We omit control-plane overheads such as overlapping metrics costs and Qd-tree construction;
for full-table repartitioning, we assume local memory is sufficient to hold the entire table.
We strengthen the baselines that require a predefined sort order (NN, DD) by providing an oracle key that is most suitable for each period: for TPC-H, \texttt{l\_shipdate} in $P3$, $P4$, \texttt{l\_commitdate} in $P5$, a Hilbert-curve combination of the two in $P6$; \review{m}{for DSB, \texttt{(ss|cs|ws)\_sold\_date\_sk}}.
We tune DD via grid search and select the optimal threshold at an overall reclustering cost comparable to WAIR, ensuring a fair, budget-matched comparison.

All experiments are conducted on the AWS cloud platform \cite{aws}, using EC2 instances as compute nodes and S3 Express One Zone as the object storage.
For TPC-H \review{m}{and DSB}, we use \texttt{m8gd.8xlarge} instances, because smaller instance types provide unstable bandwidth.
For Mirrors, we use larger \texttt{m8gd.16xlarge} instances to accommodate the increased data volume.
All data transfers between EC2 and S3 are routed through S3 gateway endpoints.

\marginnote{\ul{R1.O3}}[-1.2em]
\review{a}{
    We use CPU time as a \textit{unified metric} to assess the end-to-end cost trade-offs (discussed in \cref{sec:preliminaries}).
    We then break the total runtime into query execution time and reclustering time for detailed analysis.
    We also report the average pruning rates as a stable metric to evaluate across different cloud settings.
}

\subsection{End-to-End Results}
\label{sec:end-to-end}

\begin{figure}[t!]
    \centering
    \includegraphics[width=\linewidth]{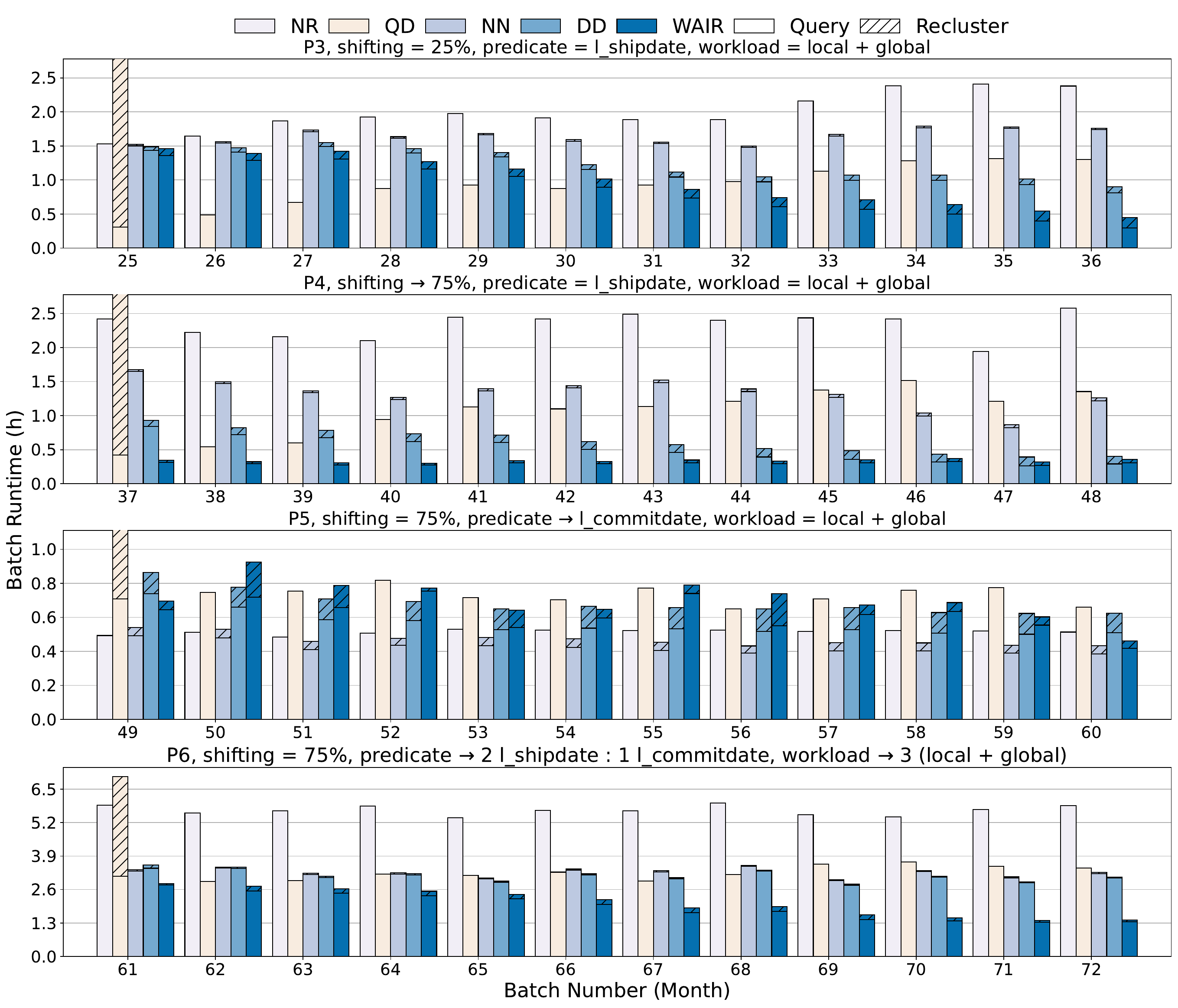}
    \caption{\review{c}{TPC-H Runtime Breakdown.}
        \textmd{Each row corresponds to a period.
            Within each row, solid bars show the monthly query cost and hatched overlays indicate the previous month's reclustering cost.
            Row titles summarize the period's primary settings.}}
    \marginnote{\ul{R3.M3}\\\ul{R3.M6}}[-3.8em]
    \Description{TPC-H Runtime Breakdown.}
    \label{fig:tpc-h}
\end{figure}

\cref{fig:overall} summarizes the end-to-end results.
Across all workloads, \algname consistently outperforms the baselines.
\reversemarginpar\marginnote{\ul{Meta}\\\ul{R1.O1}\\\ul{R1.O3}\\\ul{R4.O1}}[1.2em]
Compared to NR, \algname reduces the cumulative runtime by 61.1\% on \textbf{TPC-H},
\review{m}{60.3\% on \textbf{DSB,}} and 86.2\% on \textbf{Mirrors}.
\review{a}{The I/O volume closely matches the runtime results.}
Although QD spends a significant portion of runtime on full-table repartitioning at the beginning of each period, its query execution is still slower than \algname because QD is unable to adapt to the minor workload shifts within a period.
DD incurs almost identical reclustering overhead as \algname, but the query-time saving is much smaller due to its lack of workload awareness.
\review{m}{
    We next analyze the end-to-end results of each benchmark in detail with further breakdowns and sensitivity analysis over key parameters.
}

\subsection{TPC-H}
\label{sec:tpc-h-eval}

\subsubsection{\textbf{Per-Period Breakdown}}
\cref{fig:tpc-h} breaks down TPC-H's runtime on a per-period, per-batch basis, allowing us to trace how well each algorithm reacts to the four evaluation periods.

\textbf{Period 3 (mild drift; warm-up).}
Modest workload variation warms up the methods.
Incremental methods (NN, DD, \algname) improve steadily with low per-batch cost.
Full repartitioning (QD) pays a large one-time rewrite, starts strong, then degrades as the workload evolves.
\algname's cost model detects a highly disordered pool and invests more reclustering, yielding large immediate gains.

\textbf{Period 4 (higher variability; same predicate).}
\algname's cost model and sliding window converge, scaling back maintenance and focusing on critical partitions.
\algname sustains the best query cost with lower reclustering overhead.
QD incurs similar overhead yet its performance drops sharply, making it hard to recoup the costly repartition.
DD and NN are not workload-aware; they keep a fixed maintenance budget and deliver smaller gains than \algname.
Even when DD eventually approaches \algname's query performance, it does so with higher reclustering cost both per batch and in aggregate.

\textbf{Period 5 (predicate shift).}
The primary predicate switches to \texttt{l\_commitdate}.
Because ingestion now aligns more with the new predicate (30-90 day instead of prior 1000-day mismatch),
NR and NN benefit ``for free.''
Since QD's training did not consider the new column, its heavy rewrite yields little improvement.
Both DD and \algname recluster more in response to the shift, causing a temporary dip.
DD discards its earlier metrics and reclusters broadly, while \algname targets boundary partitions.
Guided by its cost model and sliding window, \algname adapts quickly in a sharp workload shift.

\begin{figure}[t!]
    \centering
    \includegraphics[width=\linewidth]{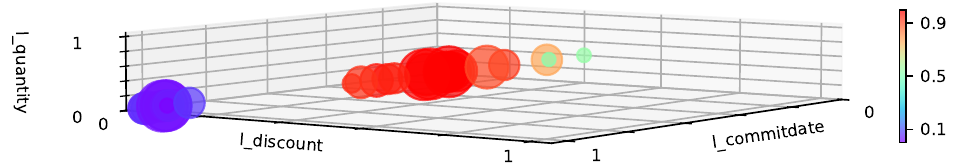}
    \caption{Centroids of Reclustered Partitions.
        \textmd{Each point denotes a centroid identified across all evaluation periods in TPC-H.
            Point size reflects the number of partitions sharing that centroid.
            Color encodes the \texttt{l\_shipdate} coordinate value.}}
    \Description{Centroids of Reclustered Partitions.}
    \label{fig:centroid}
\end{figure}

\textbf{Period 6 (mixed predicates).}
With a 2:1 mix of \texttt{l\_shipdate} and \texttt{l\_commitdate}, all methods face challenges.
QD continues to degrade.
DD and NN adopt a Hilbert-composed key but cannot adequately separate competing access patterns and serve the mixed workload, limiting gains despite reclustering cost.
\algname maintains its performance advantage by switching to a hybrid layout that quickly places centroids balancing both columns (see \cref{fig:centroid}).

\subsubsection{\textbf{Workload Characteristics}}

We then analyze the strengths and weaknesses of each method by comparing their different query costs under the recency-skewed \texttt{local} workload and the uniform \texttt{global} workload in TPC-H.
As shown in~\Cref{fig:overall}, DD's gains on \texttt{local} are smaller than its gains on \texttt{global}.
This is because DD treats the dataset uniformly without workload awareness.
It fails to prioritize recent partitions that dominate local queries.
NN achieves comparable \texttt{local} performance to DD but performs worse on \texttt{global}.
Each batch is sorted in isolation, producing many disjoint ``sorted runs'' that fragment the table and gradually degrade clustering quality.
Because NN never reclusters existing data, the initial pool remains unordered, resulting in poor performance on the global workload.
QD excels on the global workload because that workload is uniformly distributed and repeats over time, making a costly one-time full repartition beneficial.
However, a Qd-tree trained on historical statistics quickly becomes stale on \texttt{local}, which is biased toward recent data.
In contrast, \algname adapts gracefully to both workload profiles by continuously adjusting its reclustering strategy to match observed query patterns, delivering balanced and robust performance in a mixed workload.

\review{m}{
    \subsubsection{\textbf{Sensitivity Analysis}}
    \normalmarginpar\marginnote{\color{black}\ul{Meta}\\\ul{R1.O4}\\\ul{R4.O1}}

    \begin{figure}[t!]
        \centering
        \includegraphics[width=\linewidth]{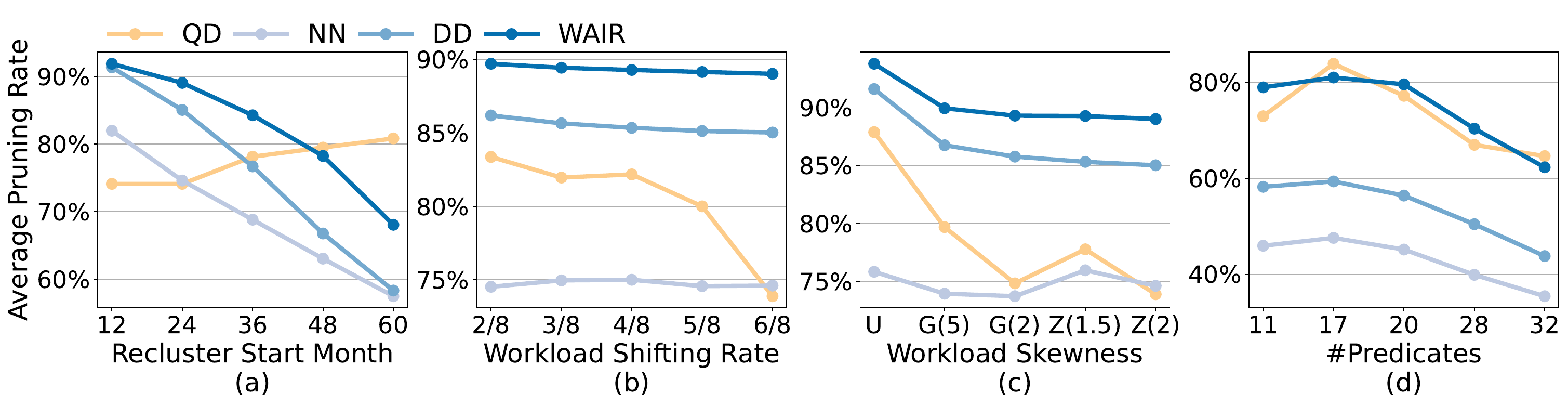}
        \vspace{-2em}
        \caption{\review{m}{Sensitivity Study.
                \textmd{Each subplot varies one key parameter while holding the others at their default values.
                    U indicates a Uniform distribution. G indicates a Gaussian distribution with $\sigma$, and Z indicates a Zipfian distribution with $\alpha$.}}}
        \marginnote{\ul{Meta}\\\ul{R1.O4}\\\ul{R4.O1}}[-3.8em]
        \Description{Sensitivity Study.}
        \label{fig:sensitivity}
    \end{figure}

    We then use our workload generator to perform controlled sensitivity analysis on the following key parameters.
    The default workload follows the same modified $\mathbf{P4}$ settings for each of the 6 periods: predicate = \texttt{l\_shipdate}, shifting rate = 75\%, workload = \texttt{local}.
    We vary one parameter at a time while holding the others at their default values.

    \textbf{Reclustering Start Time.}
    As shown in \Cref{fig:sensitivity} (a), all methods but QD perform better when reclustering starts earlier because a delayed start time leads to a larger unclustered data pool that degrades query performance.
    In contrast, because QD reclusters the entire table, it benefits slightly from a delayed start time due to more historical query data.
    \algname remains competitive even when 70\% of the data has already been ingested without clustering.

    \textbf{Workload Shifting Rate.}
    As shown in \Cref{fig:sensitivity} (b), QD's pruning rate drops dramatically as the shifting rate increases because its static partitioning quickly becomes stale.
    NN and DD are less affected because they are data-driven.
    \algname, despite its workload awareness, remains robust because its cost model tracks the shifts and adapts the reclustering accordingly.

    \textbf{Workload Skewness.}
    A skewed workload favors recently ingested data, making it harder to maintain effective clustering.
    As shown in \Cref{fig:sensitivity} (c),
    QD's pruning rate drops sharply even under mild skewness because its static partitioning cannot adapt.
    \algname exhibits slightly smaller pruning rate regression than DD under highly skewed workloads
    because \algname continuously adjusts its clustering strategy according to the workload skewness.
}

\review{m}{
    \subsection{DSB}
    \label{sec:dsb-eval}

    \marginnote{\color{black}\ul{Meta}\\\ul{R1.O1}\\\ul{R4.O1}}
    Although DSB is far more complex than TPC-H, \algname consistently outperforms the baselines as the workload evolves, as shown in \cref{fig:dsb}.
    These results validate \algname's pruning effectiveness on the complex workloads.
    Although more non-leaf operators (e.g., join, aggregate) in a query plan add compute rather than I/O complexity directly,
    they are also likely to involve more predicates and columns from different base tables (and thus more pruning opportunities).
    \algname's savings signatures and hybrid layouts can help identify the most effective base column(s) for reclustering.

    \begin{figure}[t!]
        \centering
        \includegraphics[width=\linewidth]{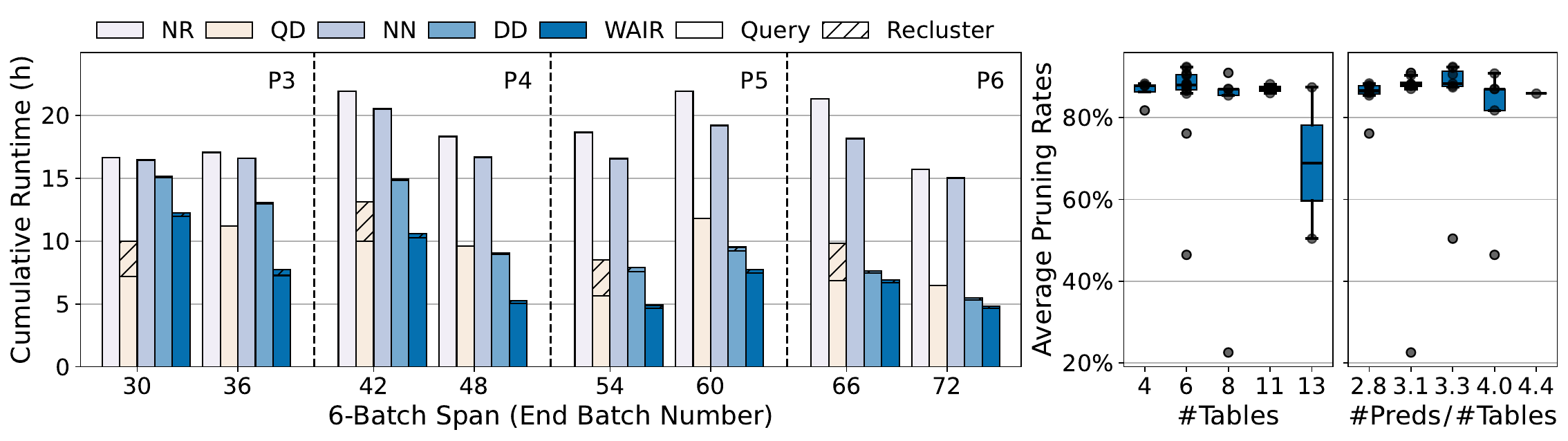}
        \vspace{-2em}
        \caption{\review{m}{DSB Breakdown.
                \textmd{Left: Cumulative 6-batch runtime.
                    Right: Box plots of query pruning rates in \algname, grouped by (1) the number of distinct base tables and (2) the number of distinct predicate attributes per table.}}}
        \marginnote{\ul{Meta}\\\ul{R1.O1}\\\ul{R4.O1}}[-3.8em]
        \Description{DSB Breakdown.}
        \label{fig:dsb}
    \end{figure}

    \Cref{fig:dsb}~(right) also reports the distribution of \algname's pruning rate per query template  grouped by
    (1) number of base tables involved, and (2) number of predicated attributes per table.
    \algname consistently achieves a high pruning rate as the query complexity increases (i.e., more joins and diverse predicates).
    The outliers (e.g., Q23) typically do not have effective range filters applied on the base tables, leaving limited pruning opportunities.
}

\marginnote{\ul{Meta}\\\ul{R1.O4}\\\ul{R4.O1}}
\review{m}{
    Leveraging our workload generator, we conduct a sensitivity analysis by varying the average number of distinct predicated columns per query.
    As shown in \cref{fig:sensitivity} (d), the pruning rates for all methods follow similar trends as more predicated columns are included in each query.
    \algname outperforms NN and DD by a large margin because \algname is able to identify the most effective base columns and apply hybrid layouts to improve pruning efficiency.
    QD also performs well in this analysis because
    there are sufficient query statistics to train a complex learned layout.
}

\subsection{Mirrors}
\label{sec:mirrors-eval}

\begin{figure}[t!]
    \centering
    \includegraphics[width=\linewidth]{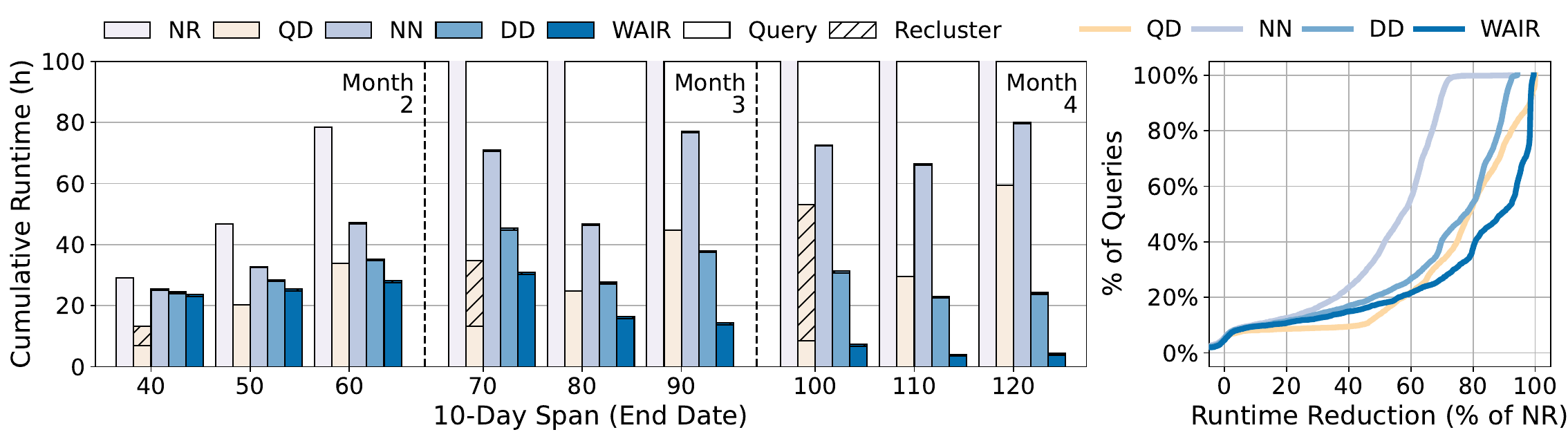}
    \vspace{-2em}
    \caption{\review{c}{Mirrors Breakdown.}
        \textmd{Left: Cumulative 10-day runtime.
            Right: ECDF of per-query runtime reduction percentages relative to NR.}}
    \marginnote{\ul{R3.M5}}[-2.5em]
    \Description{Mirrors Breakdown.}
    \label{fig:mirrors}
\end{figure}

\marginnote{\ul{Meta}\\\ul{R3.M1}}
\review{m}{
    While synthetic benchmarks provide controlled environments, the Mirrors scenario represents the chaotic reality of production cloud data warehouses.
    Derived from 2.4 TB of web access logs, this workload is characterized by \textit{extreme data skew} and \textit{highly selective access patterns} (e.g., risk checks on specific IPs or narrow time windows).
    Mirrors yields an average selectivity of only 0.1\%.
    This offers a rigorous test for reclustering decisions in ``needle-in-a-haystack'' real-world scenarios.

    As shown in \cref{fig:mirrors}, \algname's performance advantage over the baselines increases as the workload execution proceeds because \algname continuously reclusters the most critical boundary micro-partitions adaptive to the workload.
    The data-driven approach (i.e., DD), on the other hand, makes suboptimal reclustering decisions in this skewed, highly selective real workload due to the lack of workload awareness.
    QD struggles in this scenario because a static layout optimized for yesterday's traffic often fails to effectively prune today's IP queries.
    \cref{fig:mirrors} (right) shows the ECDF of each method's speedup relative to the NR baseline.
    \algname makes more than 60\% of the queries execute at least 80\% faster than the NR baseline with negligible reclustering overhead.
}

\subsection{Selectivity Analysis}
\label{sec:selectivity}

\begin{figure}[t!]
    \centering
    \includegraphics[width=0.95\linewidth]{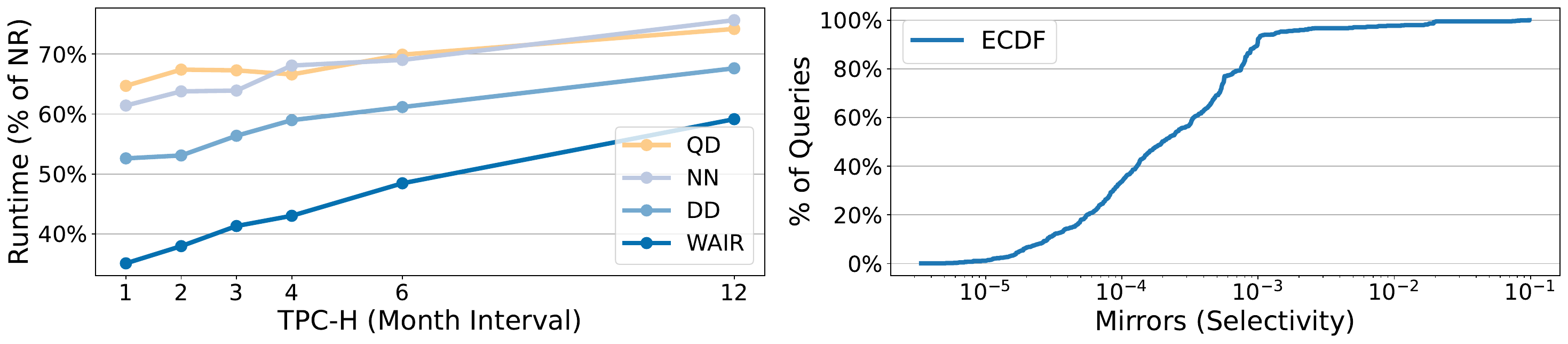}
    \caption{Selectivity Study.
        \textmd{We vary query selectivity by adjusting the TPC-H predicate window and report performance relative to NR.
            The Mirrors workload exhibits low selectivity.}}
    \Description{Selectivity Study}
    \label{fig:selectivity}
\end{figure}

\marginnote{\ul{Meta}\\\ul{R3.M1}}
\review{m}{
    We provide a sensitivity analysis on the predicate selectivity using the TPC-H workloads.
}
As shown in \cref{fig:selectivity} (left), \algname consistently outperforms all the baselines, and the performance advantages grow as the selectivity decreases.
This is because \algname's cost model adapts to selectivity to balance query and reclustering costs, while the baselines are unaware of selectivity.
\algname's dominance in the Mirrors workload is partly because of the low-selectivity nature of the queries.
\cref{fig:selectivity} (right) presents the ECDF of the selectivity of Mirrors' queries.

\subsection{Ablation Study and Optimality Gap}
\label{sec:ablation}

\begin{figure}[t!]
    \centering
    \includegraphics[width=0.95\linewidth]{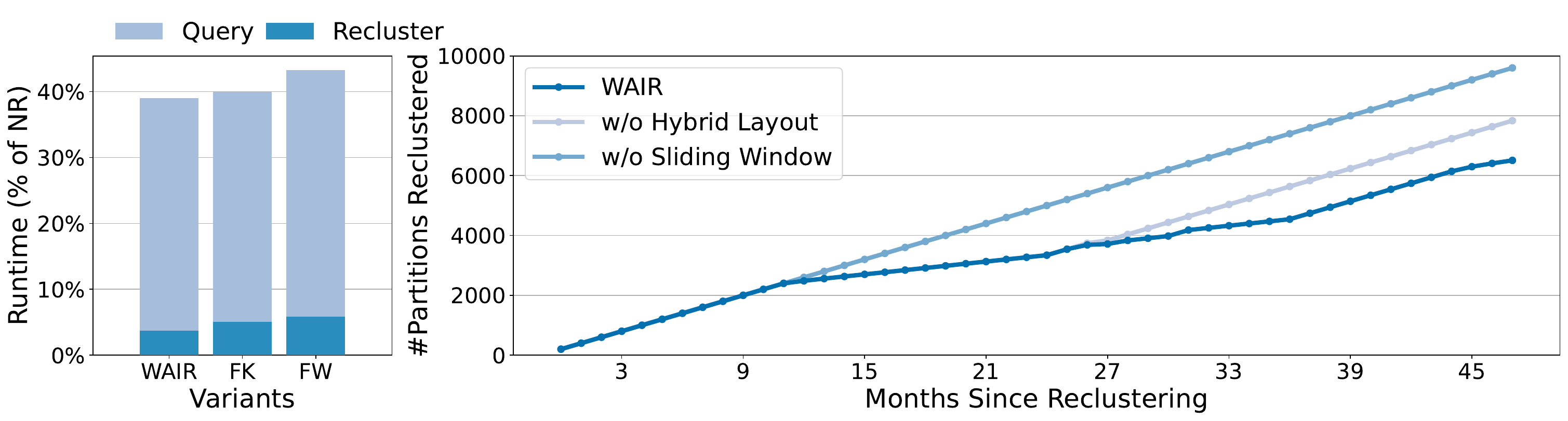}
    \caption{Ablation Study on \algname.
        \textmd{Divided runtime and reclustering intensity for \algname and its ablated \review{c}{(FK: fixed-key, FW: fixed-window)} variants on the TPC-H workload.}}
    \marginnote{\ul{R3.M4}}[-3.8em]
    \Description{Ablation Study on \algname}
    \label{fig:speed}
\end{figure}

We test \algname's resilience to continuous ingestion and workload drift by comparing it with two ablated variants.
\Cref{fig:speed} shows that \algname consistently outperforms these variants, and \algname's reclustering cost fluctuates over time as it adapts to workload shifts.
When hybrid layouts are disabled, the fixed-key policy cannot adapt to shifting predicates, resulting in degraded query performance even with higher reclustering cost.
Replacing the adaptive sliding window with a static policy reclusters the same, fixed number of partitions in each batch.
This variant reclusters far more data than necessary, incurring excessive reclustering costs with marginal query improvements.

\marginnote{\ul{R1.O4}}[-1.2em]
\review{a}{
    We evaluate \algname's ``aggressiveness'' extension (described in \cref{sec:extensions}) against its standard configuration using the TPC-H workload.
    As shown in \cref{fig:aggressive}, the aggressive variant ($\alpha=0.5$, relaxed budget $c=10$) yields greater cost reductions during early unclustered and workload-stable batches.
    However, as workload shifts intensify at $\mathbf{P5}$ and $\mathbf{P6}$, it incurs two significant cost spikes.
    The overhead from heavy reclustering is not fully amortized by subsequent savings.
    Although the aggressive variant improves overall query performance by 9.5\% with comparable total costs,
    the aggressive settings introduce non-trivial tuning complexity and excessive volatility (20\% more cost fluctuations in \cref{fig:aggressive}), which is undesirable for production environments.}

\begin{figure}[t!]
    \centering
    \includegraphics[width=0.95\linewidth]{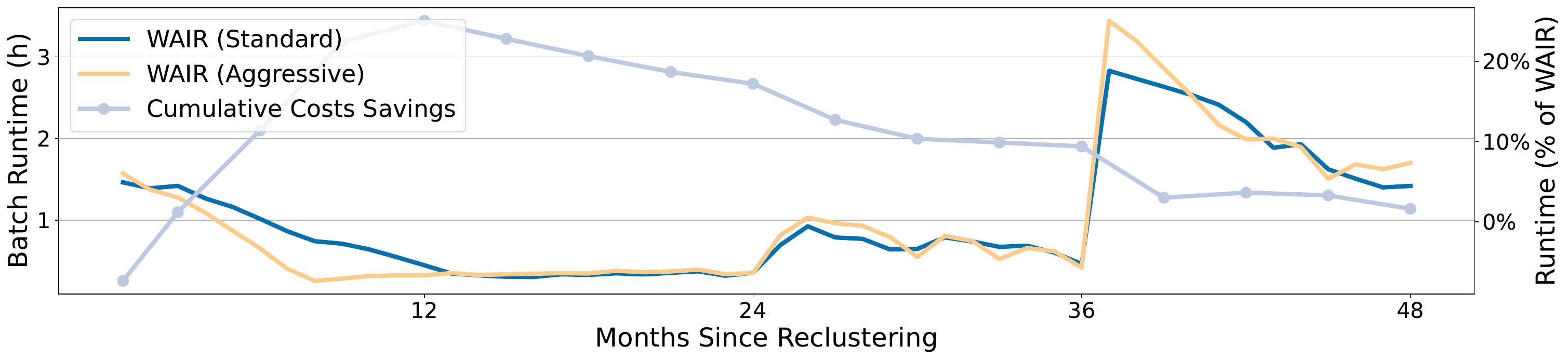}
    \caption{\review{a}{Aggressiveness Extension on \algname.
            \textmd{A more aggressive setting ($\alpha=0.5$, $c=10$) in \cref{sec:extensions} is evaluated against the standard \algname.
                The figure reports per-batch runtime and the cumulative cost reduction relative to the standard \algname.}}}
    \Description{Ablation Study on \algname}
    \marginnote{\ul{R1.O4}}[-6.25em]
    \label{fig:aggressive}
\end{figure}

\begin{figure}[t!]
    \centering
    \includegraphics[width=0.95\linewidth]{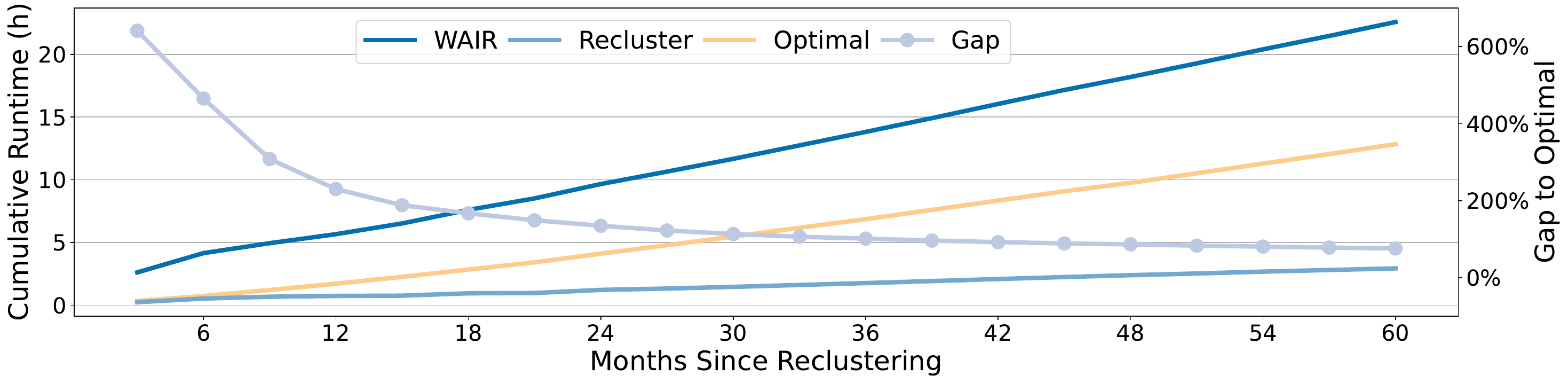}
    \caption{Optimality Gap.
        \textmd{Relative gap between \algname's cumulative cost and the oracle's query-only cost.
            Evaluated under the same TPC-H setup with a fixed predicate column (\texttt{l\_shipdate}).}}
    \Description{Optimality Gap}
    \label{fig:optimal}
\end{figure}

We further evaluate \algname against a theoretical oracle in which the table is fully sorted before query execution in \cref{fig:optimal}.
The relative gap narrows from roughly 640\% in the early batches to cumulatively 75\% by the end, at only 12.9\% overall reclustering cost.
These findings indicate that \algname achieves near-optimal query performance at a fraction of the maintenance overhead, and \algname's sliding window and hybrid layout mechanisms are essential for maintaining high performance under dynamic workloads.

\section{CONCLUSION}
This paper advocates a clean separation between \textit{reclustering policy} and \textit{clustering-key selection}, classifying prior work to highlight their inflexibility.
We formalize boundary micro-partitions and prove that reclustering them yields near-optimal pruning with bounded logarithmic amortized overhead.
Building on this, we present the \algname algorithm and implement it into a prototype automatic reclustering service.
\algname uses a sliding window of recent queries and a cloud-cost model to select high-payoff boundary partitions and reorganize them into hybrid layouts with per-group keys that maximize expected savings.
\marginnote{\ul{Meta}\\\ul{R1.O1}\\\ul{R1.O4}\\\ul{R4.O1}}
\review{m}{
    Across TPC-H, DSB, and a large real-world workload, \algname cuts total cost by up to 61\%, 60\%, and 86\%, respectively, outperforming research prototypes and documented commercial baselines.
    Breakdowns and sensitivity analyses show these gains are robust under dynamic conditions.
}
In short, \algname turns continuous, cost-effective reclustering into a practical reality for modern cloud data warehouses.

\begin{acks}
    This paper was supported by the National Natural Science Foundation of China (Grant No. 62532001), Xiongan AI Institute, and Shanghai Qi Zhi Institute.
    We would also like to thank Jiaoyi Zhang, Miao Wang, and Yihao Liu for their helpful discussions and guidance.
\end{acks}

\bibliographystyle{ACM-Reference-Format}
\bibliography{references}

\end{document}